\documentclass[11pt,letterpaper]{article}
\usepackage[letterpaper, margin=1in]{geometry}

\usepackage{graphicx} %
\usepackage{style}
\usepackage{algorithm}
\usepackage[noend]{algpseudocode}
\usepackage{mdframed}
\usepackage{ytableau}

\allowdisplaybreaks

\newcommand{\ulam}{\mathrm{\mathrm{S}}^{\mathrm{U}}}

\newcommand{\lcs}{\mathrm{LCS}}
\newcommand{\tr}{\mathrm{tr}}
\newcommand{\wvalue}{\nu}
\newcommand{\cycle}{c}
\newcommand{\dulam}{\mathrm{\mathrm{d}}^{\mathrm{U}}}
\newcommand{\dcayley}{\mathrm{\mathrm{d}}^{\mathrm{C}}}
\newcommand{\cayley}{\mathrm{\mathrm{S}}^{\mathrm{C}}}

\newcommand{\id}{\mathrm{id}}
\newcommand{\fix}{\mathrm{fix}}
\newcommand{\twocycle}{\mathrm{cycle_2}}

\newcommand{\bS}{\mathbb{S}}

\DeclareMathOperator{\Rec}{Rec}

\let\epsilon\varepsilon
\theoremstyle{definition} %
\newtheorem{example}[theorem]{Example}

\providecommand{\Comments}{1}
\ifnum\Comments=1
\usepackage[colorinlistoftodos,prependcaption,textsize=scriptsize]{todonotes}
\definecolor{Gred}{RGB}{219, 50, 54}
\definecolor{Ggreen}{RGB}{60, 186, 84}
\definecolor{Gblue}{RGB}{72, 133, 237}
\definecolor{Gyellow}{RGB}{204, 204, 0}
\definecolor{Gpurple}{RGB}{204, 0, 204}
\definecolor{Gorange}{RGB}{255, 200, 120}
\definecolor{Gbrown}{RGB}{0, 204, 204}
\paperwidth=\dimexpr \paperwidth %
\marginparwidth=\dimexpr\marginparwidth %
\else
\usepackage[disable]{todonotes}
\fi

\usepackage{colortbl}

\title{On the LSH Distortion of Ulam and Cayley Similarities}
\author{Flavio Chierichetti\thanks{Reddit, San Francisco. Work done in part while at Sapienza University of Rome. \texttt{flavio.chierichetti@reddit.com}} 
\hspace*{0.7cm}
Mirko Giacchini\thanks{Sapienza University of Rome. \texttt{\{giacchini, tani\}@di.uniroma1.it}}
\hspace*{0.7cm}
Ravi Kumar\thanks{Google, Mountain View.  \texttt{ravi.k53@gmail.com}}
\hspace*{0.7cm}
Erasmo Tani\footnotemark[2]
}
\date{}

\begin{document}
\maketitle
\begin{abstract}
Locality-sensitive hashing (LSH) has found widespread use as a fundamental primitive, particularly to accelerate nearest neighbor search. An LSH scheme for a similarity function $S:\mathcal{X} \times \mathcal{X} \to [0,1]$ is a distribution over hash functions on $\mathcal{X}$ with the property that the probability of collision of any two elements $x,y\in \mathcal{X}$ is exactly equal to $S(x,y)$. However, not all similarity functions admit exact LSH schemes. The notion of LSH distortion measures how multiplicatively close a similarity function is to having an LSH scheme. 

In this work, we study the LSH distortion of the Ulam and Cayley similarities, which are popular similarity measures on permutations of $n$ elements.  We show that the Ulam similarity admits a sublinear LSH distortion of $O(n / \sqrt{\log n})$; we also prove a lower bound of $\Omega(n^{0.12})$ on the best LSH distortion achievable. On the other hand, we show that the LSH distortion of the Cayley similarity is $\Theta(n)$.
\end{abstract}

\thispagestyle{empty}
\setcounter{page}{0}
\newpage

\section{Introduction}

Locality-sensitive hashing (LSH) is an important tool for processing
high dimensional data~\citep{b97, bcfm00, imrv97, im98}.  
Classic examples of LSH include minHash for the 
Jaccard similarity~\citep{b97, bcfm00}, simHash for the 
cosine similarity~\citep{c02}, and sampling hash for the Hamming
similarity~\citep{im98}.  LSH has been extensively used for solving 
(approximate) nearest-neighbor search~\citep{ai08} and similarity 
search~\citep{gim99, ljwcl07, sp12}; see \url{https://www.mit.edu/~andoni/LSH/} for more pointers on classical results.

LSH is a way to represent a similarity $S$---which in the context of this paper is a bi-variate, reflexive, symmetric function---by a distribution over partitions. These partitions correspond to the pre-images of hash functions, so that two entities are considered similar by a hash function if they both hash to the same value.  
In other words, an LSH family $\cH$ induces a 
\emph{collision kernel} of the form 
$S(x, y) = \Pr_{h \sim \cH} [h(x) = h(y)]$, for entities $x, y$ in the
universe.

Such collision kernels have to satisfy two constraints: 
they have unit diagonals (capturing reflexivity) and are a convex combination of equivalence-relation matrices (capturing hashing to the same value).  Hence, it turns out that many natural and useful similarity measures do not admit an LSH.  An interesting line of work studies which similarities have an LSH.  \citet{c02} identified a necessary condition for a similarity to admit an LSH: the corresponding distance function (which is often one minus the similarity) must be a metric and must be isometrically embeddable into $\ell_1$. 

To circumvent this,~\cite{ckpt19} introduced the notion of LSH distortion, which answers the question: how much is a non-LSHable similarity $S$ multiplicatively close to an LSHable similarity?

Formally, given a similarity $S$ that does not admit an LSH, is there another similarity $S'$ that has an LSH and satisfies:
\begin{equation}\label{eq:distortion-def}
    \forall x,y : \quad {1\over \Delta} \cdot S(x, y) \leq S'(x, y) \leq S(x, y)?
\end{equation}

We call the optimal $\Delta$ in \eqref{eq:distortion-def} the \emph{LSH distortion} of $S$.

In their paper, Chierichetti et al.\ showed tight LSH distortion bounds for several set similarities. Note that, since we are considering multiplicative approximations over \textbf{similarities}, and we wish to achieve them using collision kernels, the LSH distortion of a similarity can differ significantly from its metric counterpart, and known results for metric distortion do not transfer to this setting.

\paragraph{Our contributions.}
In this work, we study the LSH distortion of 
two natural similarities derived from popular metrics on permutations. The first is the normalized Ulam similarity, which measures the length of the longest common subsequence between two permutations. The second is the Cayley similarity, which is the normalized number of cycles in the relative permutation (its distance counterpart is number of transpositions needed to turn one permutation into another).

Our main results are the following:
\begin{itemize}
\item We show that the Ulam similarity has sublinear LSH distortion.  In particular, we show that it can be approximated by a similarity that has an LSH, with distortion $O(n / \sqrt{\log n})$ (\Cref{thm:ulam-upper-bound}),
\item We also show a lower bound of $\Omega(n^{0.12})$ on the LSH distortion of the Ulam similarity (\Cref{thm:ulam-lower-bound}),
\item In contrast, we show that the Cayley similarity has 
an $\Omega(n)$ lower bound on its
LSH distortion (\Cref{thm:cayley-lower-bound}), matching a simple $O(n)$ upper bound (\Cref{thm:cayley-upper-bound}).  
\end{itemize}

\paragraph{Techniques.}

The proof of the upper bound for the Ulam similarity is
based on computing a suitable measure of correlation between the given permutation and a random permutation. The collision probability of our scheme depends on how well the input permutations are correlated with a random one. Showing that the collision probability is upper bounded by $\lcs(\cdot,\cdot)/n$ is relatively easy. On the other hand, showing it is also lower bounded by the LCS (up to sublinear distortion) is more challenging.

For the lower bounds on the distortion, we crucially utilize 
the PSDness of the collision kernel (\Cref{lem:psd-Pa}).  The lower bound for the Ulam similarity is based on first constructing a small set of permutations on an 8-element
universe that have an LSH distortion of at least 9/7, and then amplifying its cardinality to work for a universe of arbitrarily large size, by using the wreath product on permutations.

The lower bound for the Cayley similarity is based on representation theory, and uses both the fact that every bi-invariant PSD kernel on $\mathbb{S}_n$ can be written as a convex combination of normalized characters, as well as Roichman’s character-ratio estimates applied to derangements.
The distortion lower bound is obtained by constructing 
two derangements that share the high-order terms of this decomposition and hence must have close collision probabilities with the identity permutation, but have vastly different Cayley similarities (linear vs constant) to it.

\paragraph{Related work.}

For permutation metrics, the Ulam metric has been extensively studied through its connection to edit distance and longest common subsequences.  \cite{ck06} proved that the Ulam metric embeds into $\ell_1$ with $O(\log n)$ distortion;~\citet{ak10} later proved that embeddings of the Ulam metric into $\ell_1$ requires a near-tight 
distortion of $\Omega(\log n/\log\log n)$.  
\citet{aik09} bypassed the $\ell_1$ barrier by embedding Ulam distance with constant distortion into an iterated product metric. \citet{cgkk18} obtained an $O(\log n)$ distortion randomized embedding of Ulam into Hamming space that also carries alignment information.  Note that all these previous results are for the Ulam metric; to the best of our knowledge, the corresponding question of Ulam similarity has not  been considered in the literature. 

Unlike Ulam, the Cayley metric is not hard to embed into $\ell_1$. Our work can be viewed as a hardness result for LSHability of the Cayley similarity. 

\paragraph{Outline.} We begin, in \Cref{sec:notation}, by introducing some basic notation. In \Cref{sec:lsh-intro}, we review the definition of LSH distortion and introduce in detail the central problems of the paper. In \Cref{sec:ulam-ub}, we provide our $O(n / \sqrt{\log n})$ upper bound on the LSH distortion for the Ulam similarity. We then give a $\Omega(n^{0.12})$ lower bound in \Cref{sec:ulam-lower-bound}. In \Cref{sec:cayley-upper-bound}, we give a simple $O(n)$ upper bound on the distortion of the Cayley similarity, that we complement with a matching lower bound in \Cref{sec:cayley-lower-bound}.

In \Cref{app:background-representation-theory} we also review some background on the representation theory of $\mathbb{S}_n$, which we rely on for our proofs in \Cref{sec:cayley-lower-bound}.

\section{Notation}\label{sec:notation}

\paragraph{Permutations.} For a positive integer $n$, let $[n]:=\{1,\dots, n\}$. Let $\mathbb{S}_n$ be the set of all permutations of $[n]$. We denote by $\id$ the identity permutation. For a permutation $\pi \in \mathbb{S}_n$ we will sometimes use the one-line notation and sometimes the cycle decomposition.\footnote{For example, for the permutation $\pi:[5]\rightarrow[5]$ such that $\pi(1)=3$, $\pi(2)=5$, $\pi(3)=1$, $\pi(4)=2$, $\pi(5)=4$; the one-line representation is $(3,5,1,2,4)$ and the cycle decomposition is $(1\,\, 3)(2\,\,5\,\,4)$.} For a permutation $\pi\in \mathbb{S}_n$ we denote with $c(\pi)$ the number of cycles in its cycle decomposition. Note that $\cycle(\pi)=\cycle(\pi^{-1})=\cycle(\sigma \pi \sigma^{-1})$ for any $\sigma\in \mathbb{S}_n$. For $\pi\in \mathbb{S}_n$, $a,b\in [n]$ we write $a <_\pi b$ to indicate that $a$ comes before $b$ in the one-line representation of $\pi$ (alternatively that $\pi^{-1}(a) < \pi^{-1}(b))$. We write $a \leq_\pi b$ to indicate that either $a=b$ or $a<_\pi b$. For $\pi,\sigma\in\mathbb{S}_n$ we denote with $\lcs(\pi,\sigma)$ the length of the longest common subsequence between the one-line representations of $\pi$ and $\sigma$. For $\pi\in\mathbb{S}_n$, we denote with $\fix(\pi)$ the number of fixed points of $\pi$ and with $\twocycle(\pi)$ the number of cycles of length 2 in the cycle decomposition of $\pi$. We also let $\supp(\pi)=n-\fix(\pi)$.

\paragraph{Partitions and Compositions.} Recall that, for any positive integer $n$, a partition $\lambda =(\lambda_1, \ldots, \lambda_r)$ of $n$ is a finite sequence of non-increasing positive integers that sum to $n$. We write $\lambda \vdash n$ to indicate that $\lambda$ is a partition of $n$. We denote by $(1^n)$ the partition $(1,\ldots ,1)$ containing $n$ $1$s. Given a partition $\lambda =(\lambda_1, \ldots, \lambda_r) \vdash n$ its \emph{transpose partition} (sometimes also called the \emph{conjugate partition}) is the partition $\lambda '\vdash n$ defined by $\lambda'=(\lambda'_1, \ldots , \lambda'_{\lambda_1})$, where:
\[
    \forall i \in [\lambda_1] : \lambda'_i = \max\{k \mid \lambda_k \ge i\}.
\]
 A \emph{composition} of $n$ is an ordered partition of $n$, i.e.\ an integer vector $\mu = (\mu_1,\dots, \mu_\ell)$ such that $\mu_i \ge 0$ for every $i\in[\ell]$ and  $\sum_{i\in[\ell]} \mu_i =n$.

\paragraph{Matrices.} A symmetric matrix $M \in \mathbb{R}^{n \times n}$ is positive semidefinite (PSD, write $M\succeq 0$) if for all $x\in \mathbb{R}^n$, we have:
\[
    x^TMx\ge 0.
\]
The trace operator is the function that takes in a square matrix $M \in \mathbb{R}^{n\times n}$ and returns $\tr (M) = \sum_{i\in[n]} M(i,i)$. Given two square matrices $X,Y$, we denote by $\operatorname{diag}(X,Y)$ the block matrix:
\[
    \operatorname{diag}(X,Y) = \begin{pmatrix}
        X & \mathbf{0}\\
        \mathbf{0} & Y
    \end{pmatrix}
\]
where $\mathbf{0}$ is the all-zeros matrix.

\section{Similarities, LSH, and Distortion}\label{sec:lsh-intro}
\begin{definition}[Similarity]
Given a set $\mathcal{X}$, a \emph{similarity} (or \emph{similarity function}) on $\mathcal{X}$ is a function $S:\mathcal{X}\times \mathcal{X} \to [0,1] $, satisfying:
\begin{description}
    \item[(a) Reflexivity:] $S(x,x) = 1$ for all $x\in \mathcal{X}$,
    \item[(b) Symmetry:] $S(x,y) = S(y,x)$ for all $x,y\in \mathcal{X}$.
\end{description}
A \emph{similarity space} is simply a pair $(\mathcal{X},S)$, where $\mathcal{X}$ is a set, and $S$ is a similarity function on $\mathcal{X}$.
\end{definition}

Given a similarity space, it is natural to ask whether one can find a small representation for elements of $\mathcal{X}$ that preserve the values of $S(\cdot,\cdot)$. A natural approach to obtaining this is locality-sensitive hashing, which aims to construct randomized embeddings of $\mathcal{X}$ into discrete ($0$-$1$) similarity spaces that preserve $S(\cdot,\cdot)$ in expectation. We recall the key definition here.

\begin{definition}[Locality-Sensitive Hashing (LSH)]
    A locality-sensitive hashing (LSH) scheme for a similarity space $(\mathcal{X},S)$ is a distribution $\mathcal{H}$ over functions from $\mathcal{X}$ into some finite set $\mathcal{Y}$ with the property that:
    \[
        \Pr_{h\sim \mathcal{H}}\left[h(x) = h(y)\right] = S(x,y),
    \]
    for all $x,y \in \mathcal{X}$.
\end{definition}
We say that a similarity function $S$ is \emph{LSHable} if it admits an  LSH scheme. 

\cite{ckpt19} give a notion of distortion for similarity functions, which quantifies the degree to which non-LSHable similarities are far from being LSHable.  We recall the key definitions.
\begin{definition}[\cite{ckpt19}]\label{def:LSH-distortion}
    We say a distribution $\cH$ over functions from $\mathcal{X}$ to some finite set $\mathcal{Y}$ $\Delta$\emph{-approximates} the space $(\mathcal{X}, S)$, if:
    \[
        {1\over \Delta} \cdot S(x,y)\le \Pr_{h\sim \mathcal{H}}\left[h(x) = h(y) \right] \le S(x,y),
    \]
    for all $x,y \in \mathcal{X}$.
    The LSH \emph{distortion} (or simply distortion) of a similarity space $(\mathcal{X},\mathcal{S})$ is the smallest $\Delta \ge1 $ such that there exists a distribution $\mathcal{H}$ that $\Delta$-approximates $(\mathcal{X}, S)$.
\end{definition}

A simple argument shows that the distortion of a similarity space $(\mathcal{X},S)$ could be characterized as the smallest $\Delta \ge1 $ such that there exists an LSHable similarity function $S':\mathcal{X} \times \mathcal{X}\to[0,1]$ with the property that:
\begin{equation}\label{eq:similarity-distorts}
    {1\over \Delta} \cdot S(x,y) \le S'(x,y)\le S(x,y),
\end{equation}
for all $x,y\in \mathcal{X}$. %

The focus of this work is the distortion of two similarity spaces in which $\mathcal{X}$ is the symmetric group $\mathbb{S}_n$: the Ulam similarity and the Cayley similarity.

\paragraph{Ulam Similarity.} The \emph{Ulam distance} between two permutations $\pi,\sigma\in\mathbb{S}_n$ is defined as 
\[
    \dulam(\pi,\sigma):=n-\lcs(\pi,\sigma),
\]
We define the (normalized) \emph{Ulam similarity} in a natural way:
\[
    \ulam(\pi,\sigma)= 1- \frac{n-\lcs(\pi,\sigma)}{n}=\frac{\lcs(\pi,\sigma)}{n}.
\]

\paragraph{Cayley Similarity.} 

The \emph{Cayley distance} between two permutations $\pi, \sigma\in \mathbb{S}_n$ is defined as the minimum number of (not necessarily adjacent) transpositions required to transform one permutation into the other, and it is also equivalent to $\dcayley(\pi,\sigma):=n - \cycle(\pi^{-1} \sigma)=n - \cycle(\pi \sigma^{-1})$ (see, e.g., \citep[page 118]{diaconis88}). We define a notion of (normalized) \emph{Cayley similarity}
in a natural way:
\[
    \cayley(\pi, \sigma)=1- \frac{\dcayley(\pi,\sigma)}{n}=\frac{\cycle(\pi^{-1}\sigma)}{n}.
\]

\paragraph{Positive Semi-Definiteness.} In order to prove lower bounds on the achievable distortion, we will make use of the following result. Let $\mathcal{H}$ be a probability distribution over functions from a finite domain $\mathcal{X}$ into a finite range $\mathcal{Y}$. Fixing any ordering of the elements of $\mathcal{X}$, we define matrix $P_{\cH}\in \mathbb{R}^{|\mathcal{X}|\times |\mathcal{X}|}$ such that:
\[
    P_{\cH}(x_1,x_2) := \Pr_{h\sim \mathcal{H}}\left[ h(x_1) = h(x_2)\right],
\]
for $x_1,x_2\in \mathcal{X}$.

We now show that this matrix is always PSD.
\begin{lemma}\label{lem:psd-Pa}
For any distribution $\cH$ over functions from $\mathcal{X}$ to $\mathcal{Y}$, where $\mathcal{X}$ and $\mathcal{Y}$ are finite sets, the matrix $P_{\cH}$ is PSD. That is, for all $\{w_x\}_{x\in \mathcal{X}} \in \mathbb{R}^{|\mathcal{X}|}$, we have: 
\[
    \sum_{x_1\in \mathcal{X}} \sum_{x_2\in \mathcal{X}} w_{x_1} w_{x_2} P_{\cH}(x_1, x_2) \ge 0.
\]
\end{lemma}
\begin{proof}
Note that for any $x_1,x_2\in \mathcal{X}$,
\[
    P(x_1, x_2) 
    = \Pr_{h \sim \mathcal{H}} [h(x_1) = h(x_2)]
    = \E{h \sim \cH}{\mathbf{1}_{[h(x_1) = h(x_2)]} },
\]
where $\mathbf{1}_{[\cdot]}$ is the indicator function, and hence:
\[
  \sum_{x_1\in \mathcal{X}} \sum_{x_2\in \mathcal{X}} w_x w_y P(x_1, x_2) =
  \sum_{x_1\in \mathcal{X}} \sum_{x_2\in \mathcal{X}} w_x w_y \E{h \sim \cH}{\mathbf{1}_{[h(x_1) = h(x_2)]} }
  = \E{h \sim \cH}{ \sum_{x_1\in \mathcal{X}} \sum_{x_2\in \mathcal{X}} w_{x_1} w_{x_2} \mathbf{1}_{[h(x_1) = h(x_2)]} }.
\]
Now, for a fixed  $h:\mathcal{X}\rightarrow \mathcal{Y}$, we show that the inner sum is non-negative, which in turns completes the proof. We have:
\[
    \sum_{x_1\in \mathcal{X}} \sum_{x_2\in \mathcal{X}} w_{x_1} w_{x_2} \mathbf{1}_{[h(x_1) = h(x_2)]} = \sum_{y\in \mathcal{Y}} \sum_{\substack{x_1\in \mathcal{X} \, :\\h(x_1)=y}} \sum_{\substack{x_2\in \mathcal{X} \, :\\h(x_2)=y}} w_{x_1} w_{x_2} = \sum_{y\in \mathcal{Y}} \left(\sum_{\substack{x\in \mathcal{X} \, :\\h(x)=y}} w_{x}\right)^2 \geq 0. \qedhere
\]
\end{proof}

\section{Upper Bound for Ulam Similarity}\label{sec:ulam-ub}
In this section, we prove the following result:
\begin{theorem}\label{thm:ulam-upper-bound}
    The LSH distortion of $(\mathbb{S}_n, \ulam)$ is $O(n/\sqrt{\log n})$.
\end{theorem}
For two permutations $\tau,\pi \in \mathbb{S}_n$ and $z\in [n]$, define:
\[
    \Rec_{\tau, z}(\pi):= \{a\in [n] \mid z \leq_\pi a \text{ and } b <_\tau a \text{ for all $b$ such that }z \leq_\pi b <_\pi a \}.
\]
Intuitively, $\Rec_{\tau, z}(\pi)$ is the set of new record maxima according to $\tau$ obtained while scanning the one-line representation of $\pi$ starting from $z$.

For $\tau, \pi \in \mathbb{S}_n$, $a, z \in [n]$ define:
\[
    h_{\tau, z, a}(\pi) := \begin{cases}
        1 & \text{ if }\, a\in \Rec_{\tau, z}(\pi), \\
        \pi & \text{ otherwise.} 
    \end{cases}
\]
Let $\cH$ be the uniform distribution over the set of functions: $\left\{h_{\tau, z, a} \mid \tau\in \mathbb{S}_n, z\in[n], a\in[n] \right\}$. Recall that for $\pi,\sigma\in\mathbb{S}_n$:
\[
    P_\cH(\pi, \sigma):=\Pr_{h_{\tau,z,a}\sim \cH}[h_{\tau,z,a}(\pi)=h_{\tau,z,a}(\sigma)].
\]
We will show that: $O(\sqrt{\log n}/n) \cdot \ulam(\pi,\sigma)\leq P_\cH(\pi,\sigma)\leq \ulam(\pi,\sigma)$. 

\begin{lemma}\label{lem:ub-ulam-H-exp}
    For $\pi,\sigma\in \mathbb{S}_n$, $\pi\neq\sigma$, it holds that:
    \[
        P_{\cH}(\pi,\sigma)=\frac{1}{n} \cdot \E{\tau,z}{\big|\Rec_{\tau,z}(\pi) \cap \Rec_{\tau,z}(\sigma)\big|}.
    \]
\end{lemma}
\begin{proof}
    Note that for $\pi\neq \sigma$, we have that $h_{\tau, z,a}(\pi)=h_{\tau, z, a}(\sigma)$ if and only if $a\in \Rec_{\tau,z}(\pi) \cap \Rec_{\tau,z}(\sigma)$. Conditioned on the choice of $\tau$ and $z$, there are $n$ possible values for $a$ and only $\big|\Rec_{\tau,z}(\pi) \cap \Rec_{\tau,z}(\sigma)\big|$ of them make the hashes of $\pi$ and $\sigma$ equal. This concludes the proof. 
\end{proof}

\begin{lemma}\label{lem:ub-ulam-bound-LCS}
    For any $\tau, \pi, \sigma \in \mathbb{S}_n$ and $z\in [n]$, it holds that:
    \[
        1 \leq \big|\Rec_{\tau,z}(\pi) \cap \Rec_{\tau,z}(\sigma)\big| \leq \lcs(\pi, \sigma).
    \]
\end{lemma}
\begin{proof}
    Since $z\in \Rec_{\tau,z}(\pi) \cap \Rec_{\tau,z}(\sigma)$, the set $\Rec_{\tau,z}(\pi) \cap \Rec_{\tau,z}(\sigma)$ is always non-empty, proving the first inequality. 
    
    We now show the second inequality. If there is only one element in the intersection, the result is trivially true, since $\lcs(\pi,\sigma)\ge 1$ for any $\pi,\sigma \in \mathbb{S}_n$. Otherwise, consider any distinct $x,y\in \Rec_{\tau,z}(\pi) \cap \Rec_{\tau,z}(\sigma)$. Suppose without loss of generality that $x <_\pi y$. By definition of $\Rec_{\tau, z}(\pi)$, since $y$ is observed after $x$ in the one-line representation of $\pi$, it must hold that $x <_\tau y$. Moreover, we have $x <_\sigma y$, for otherwise we would have $x>_\tau y$, leading to a contradiction. 

    Therefore, if we let $\Rec_{\tau,z}(\pi) \cap \Rec_{\tau,z}(\sigma) = r_1 <_\pi \dots <_\pi r_t$, we also have $r_1 <_\sigma \dots <_\sigma r_t$, meaning that $r_1, \dots, r_t$ is a common subsequence of $\pi$ and $\sigma$, and therefore $t\leq \lcs(\pi, \sigma)$.
\end{proof}

Fixed $\pi, \sigma$ such that $\lcs(\pi,\sigma)\geq 2$, define $L:=\lcs(\pi,\sigma)$ and let $C=(c_1, \dots, c_L)$ be a longest common subsequence. For $1 \leq i < j \leq L$, define the intervals: 
\begin{align*}
    I^\pi_{i,j}:=\{x\in[n] \mid c_i\le_\pi x\le_\pi c_j\}, \quad \text{and}\quad I^\sigma_{i, j}:=\{x\in[n] \mid c_i\le_\sigma x\le_\sigma c_j\}.
\end{align*}
Define also $W_{i,j}:=|I_{i,j}^\pi \cup I_{i,j}^\sigma|$.

\begin{lemma}\label{lem:ub-ulam-LB-H}
    For any $\pi,\sigma\in\mathbb{S}_n$ such that $L:=\lcs(\pi,\sigma) \geq 2$, it holds that:
    \[
        P_{\cH}(\pi, \sigma)\geq \frac{1}{n^2} \cdot \sum_{1\leq i < j \leq L} \frac{1}{W_{i,j}}.
    \]
\end{lemma}
\begin{proof}
    Consider the case $\pi\neq \sigma$, otherwise the result is trivial, since the right-hand side is always less than $1$. For $\tau\in \mathbb{S}_n$, and $1\leq i < j \leq L$, suppose that:
    \begin{align*}
        c_j \ge_\tau x &\quad \text{ for all } x\in I_{i,j}^\pi \cup I_{i,j}^\sigma.
    \end{align*}
    Under this assumption, it holds that: $c_j \in \Rec_{\tau, c_i}(\pi) \cap \Rec_{\tau, c_i}(\sigma)$. Indeed, we have that $c_j$ is the maximum (according to $\tau$) of all indices $I_{i,j}^\pi$ between $c_i$ and $c_j$ in the one-line representation of $\pi$ and therefore it will be added to $\Rec_{\tau, z}(\pi)$ when starting from $z=c_i$ (and the same argument holds for $\sigma$).

    By \Cref{lem:ub-ulam-H-exp} and the remark above, we have:
    \begin{align*}
        P_{\cH}(\pi, \sigma)&= \frac{1}{n} \cdot \E{\tau, z}{|\Rec_{\tau,z}(\pi) \cap \Rec_{\tau,z}(\sigma)|} \\
        &\geq \frac{1}{n} \cdot \sum_{i=1}^{L-1} \Pr[z=c_i] \cdot \E{\tau}{|\Rec_{\tau,z}(\pi) \cap \Rec_{\tau,z}(\sigma)| \,\Big|\, z=c_i}\\
        & = \frac{1}{n} \cdot  \sum_{i=1}^{L-1} \frac{1}{n} \cdot \sum_{j=1}^n \Pr_{\tau}[j \in \Rec_{\tau,c_i}(\pi) \cap \Rec_{\tau,c_i}(\sigma)] \\
        & \geq \frac{1}{n^2} \sum_{i=1}^{L-1} \sum_{j=i+1}^L \Pr_{\tau}[c_j \in \Rec_{\tau,c_i}(\pi) \cap \Rec_{\tau,c_i}(\sigma)] \\
        & \geq \frac{1}{n^2} \sum_{i=1}^{L-1} \sum_{j=i+1}^L \Pr_{\tau}\left[c_j \geq_\tau x \text{ for all } x\in I_{i,j}^\pi \cup I_{i,j}^\sigma \right] \\
        & = \frac{1}{n^2} \sum_{i=1}^{L-1} \sum_{j=i+1}^L \frac{1}{W_{i,j}},
    \end{align*}
    where the last inequality follows from the remark above and the last equality uses that for a uniform at random $\tau\in \mathbb{S}_n$, the element ranked highest among $I_{i,j}^\pi \cup I_{i,j}^\sigma$ is chosen uniformly at random among its elements.
\end{proof}

\begin{lemma}\label{lem:ub-ulam-Wij-LB}
    For any $\pi,\sigma\in\mathbb{S}_n$ such that $L:=\lcs(\pi,\sigma) \geq 4$, it holds that:
    \[
        \sum_{1\leq i < j \leq L} \frac{1}{W_{i,j}} \geq \frac{L^2 \ln L}{32n}.
    \]
\end{lemma}
\begin{proof}
    For $1 < t \leq L$, define:
    \[
        w_t:= \Big|I^\pi_{t-1,t}\Big| + \Big|I^\sigma_{t-1,t} \Big|.
    \]
    For $1\leq i < L$ and $1\leq d \leq L-i$, define $S_{i,d} := \sum_{t=i+1}^{i+d} w_t$. Note that, for $1\leq i < j \leq L$, we have:
    \begin{equation}\label{eq:ulam-ub-wij-ub} 
        W_{i,j} \leq |I_{i,j}^\pi| +|I_{i,j}^\sigma| \leq \sum_{t=i+1}^{j} w_t = S_{i,(j-i)}.
    \end{equation}
    Therefore, we have:
    \begin{equation}\label{eq:ulam-ub-rewrite-Wij}
        \sum_{1\leq i < j \leq L} \frac{1}{W_{i,j}} = \sum_{i=1}^{L-1}\sum_{j=i+1}^L \frac{1}{W_{i,j}} \overset{\eqref{eq:ulam-ub-wij-ub}}{\geq} \sum_{i=1}^{L-1}\sum_{j=i+1}^L \frac{1}{S_{i,(j-i)}} = \sum_{i=1}^{L-1} \sum_{d=1}^{L-i} \frac{1}{S_{i,d}} = \sum_{d=1}^{L-1} \sum_{i=1}^{L-d} \frac{1}{S_{i,d}}.
    \end{equation}
    For a fixed $1 \leq d < L$, by the Cauchy--Schwarz inequality, we have:
    \begin{equation}\label{eq:ulam-ub-cs}
        \left(\sum_{i=1}^{L-d} \frac{1}{S_{i,d}}\right) \cdot \left(\sum_{i=1}^{L-d} {S_{i,d}}\right) \geq \left(\sum_{i=1}^{L-d} \sqrt{\frac{1}{S_{i,d}}} \cdot \sqrt{S_{i,d}}\right)^2 = (L-d)^2.
    \end{equation}
    Note that each element of $[n]$ can be counted at most twice in $w_t$ and it can appear in at most two $w_t$'s, therefore:
    \[
        \sum_{t=2}^L w_t \leq 4n.
    \]
    Moreover, each $w_t$ can appear in at most $d$ distinct $S_{i,d}$ for $i\in \{1, \dots, L-d\}$, giving:
    \[
        \sum_{i=1}^{L-d} S_{i,d} \leq d \sum_{t=2}^L w_t \leq 4\cdot d \cdot n.
    \]
    Combining this inequality with \eqref{eq:ulam-ub-cs} gives:
    \[
        \sum_{i=1}^{L-d} \frac{1}{S_{i,d}} \geq \frac{(L-d)^2}{\sum_{i=1}^{L-d} S_{i,d}} \geq \frac{(L-d)^2}{4\cdot d \cdot n}.
    \]
    Plugging this inequality into \eqref{eq:ulam-ub-rewrite-Wij} we obtain:
    \begin{align*}
        \sum_{1 \leq i < j \leq L} \frac{1}{W_{i,j}} \geq \sum_{d=1}^{L-1} \frac{(L-d)^2}{4\cdot d \cdot n} \geq \sum_{d=1}^{\floor{L/2}} \frac{(L-d)^2}{4\cdot d \cdot n} \geq \frac{L^2}{16n} \sum_{d=1}^{\floor{L/2}} \frac{1}{d} \geq \frac{L^2 \ln L}{32 n},
    \end{align*}
    where the last inequality uses that for $L\geq 4$:
    \[
        \sum_{d=1}^{\floor{L/2}}\frac{1}{d} \geq \ln(\floor{L/2} + 1) \geq \ln\left(\sqrt{L}\right) = \frac{\ln L}{2}. \qedhere
    \]
\end{proof}

We are now ready to prove the main result of this section:
\begin{proof}[Proof of \Cref{thm:ulam-upper-bound}]
    It is sufficient to show that for each $\pi, \sigma\in\mathbb{S}_n$ it holds that: 
    \[
        \Theta(\sqrt{\log n} / n) \cdot \ulam(\pi, \sigma) \leq P_\cH(\pi, \sigma) \leq \ulam(\pi, \sigma).
    \]
    Note that if $\pi=\sigma$ then $P_\cH(\pi,\sigma)=\ulam(\pi,\sigma)=1$. Therefore, consider the case where $\pi\neq \sigma$. By combining \Cref{lem:ub-ulam-H-exp} and \Cref{lem:ub-ulam-bound-LCS} we have:
    \begin{equation}\label{eq:ulam-ub-H-bounds-1}
        \frac{1}{n} \leq P_\cH(\pi, \sigma) \leq \frac{\lcs(\pi, \sigma)}{n} = \ulam(\pi, \sigma). 
    \end{equation}
    Therefore, we only have to prove $P_\cH(\pi, \sigma) \geq \Omega(\sqrt{\log n}/n) \cdot \ulam(\pi, \sigma)$ to conclude the proof. Consider first the case where $\lcs(\pi, \sigma) \leq 3$. In this case we have:
    \[
        P_\cH(\pi, \sigma) \geq \frac{1}{n} \geq \frac{1}{3} \cdot \frac{\lcs(\pi, \sigma)}{n} = \frac{1}{3} \cdot \ulam(\pi, \sigma).
    \]
    Consider now the case $\lcs(\pi, \sigma)\geq 4$ and denote $L:=\lcs(\pi, \sigma)$. By combining \Cref{lem:ub-ulam-LB-H} and \Cref{lem:ub-ulam-Wij-LB} we obtain:
    \begin{equation}\label{eq:ulam-ub-H-bounds-2}
        P_\cH(\pi, \sigma) \geq \frac{L^2 \ln L}{32\cdot n^3} = \frac{L \ln L}{32 \cdot n^2} \cdot \ulam(\pi, \sigma).
    \end{equation}
    By \eqref{eq:ulam-ub-H-bounds-1} we also have:
    \begin{equation}\label{eq:ulam-ub-H-bounds-3}
        P_\cH(\pi, \sigma) \geq \frac{1}{L} \cdot \ulam(\pi, \sigma).
    \end{equation}
    We split the case $\lcs(\pi, \sigma)\geq 4$ in two sub-cases. Suppose first that $\lcs(\pi, \sigma) \leq \frac{n}{\sqrt{\ln n}}$. Then, \eqref{eq:ulam-ub-H-bounds-3} gives: $P_\cH(\pi, \sigma) \geq \sqrt{\ln n}/ n \cdot \ulam(\pi, \sigma)$. Suppose instead that $\lcs(\pi, \sigma) > \frac{n}{\sqrt{\ln n}}$. Then, by \eqref{eq:ulam-ub-H-bounds-2} we have:
    \[
        P_\cH(\pi, \sigma) \geq \frac{ \frac{n}{\sqrt{\ln n}} \ln \left(\frac{n}{\sqrt{\ln n}}\right) }{32 n^2} \cdot \ulam(\pi, \sigma) = \left( \frac{\ln n - \frac12 \ln\ln n}{32\cdot \sqrt{\ln n} \cdot n} \right) \cdot \ulam(\pi, \sigma) \geq \frac{\sqrt{\ln n}}{64 \cdot n} \cdot \ulam(\pi, \sigma),
    \]
    where we used that $\frac{\ln n}{2} \geq \frac{\ln \ln n}{2}$ for $n\geq 3$.
\end{proof}

\section{Lower Bound for Ulam Similarity}\label{sec:ulam-lower-bound}

In this section, we prove the following result.

\begin{restatable}{theorem}{UlamLowerBound}\label{thm:ulam-lower-bound}
    Let $n$ be any integer of the form $8^{2^k}$, for integer $k\geq 0$. The LSH distortion of $(\mathbb{S}_n,\ulam)$ 
    is $\Omega(n^\delta)$, where $\delta \ge 0.12$.
\end{restatable}

At a high level, we obtain a lower bound on the LSH distortion by starting with a small instance that exhibits constant distortion, and then iteratively amplifying it to handle larger values of $n$. In each step, both the distortion bound and the permutation domain grow quadratically. Our base case has $n=8$, therefore our construction produces instances for any $n$ of the form $8^{2^k}$, where $k$ is a non-negative integer. The main tool enabling this amplification is the wreath product.

For any two subsets $A,B \subseteq \mathbb{S}_n$, we define the similarity matrix $M_{A,B} \in \mathbb{R}^{|A|\times |B|}$ by $M_{A,B}(\pi,\sigma) = \ulam(\pi,\sigma)$ for any $\pi \in A$ and $\sigma \in B$. For simplicity, when $A$ and $B$ are the same, we will adopt a lighter notation, and denote by $M_A$ the matrix $M_{A,A}$. We note that the matrix $M_{(A\cup B)}$ has a simple block structure:
\[
    M_{(A\cup B)} =
    \begin{pmatrix}
        M_{A} & M_{A,B} \\
        M_{B,A} & M_{B}
    \end{pmatrix},
\]
where we fixed an ordering in which the permutations in $A$ come before those in $B$.

\begin{definition}[Witness]\label{def:ulam-witness}
    For $n> 0$, let $A,B\subseteq \mathbb{S}_n$ be disjoint subsets such that $|A|=|B|$ and $M_{A}=M_{B}$. Let $W\in \mathbb{R}^{(|A|+|B|) \times (|A|+|B|)}$ be a PSD block matrix such that:
    \[
        W = \begin{pmatrix} U & -V \\ -V^T & U \end{pmatrix} \succeq 0,
    \]
    where $U,V\in \mathbb{R}^{|A|\times |A|}$ are entry-wise non-negative symmetric matrices such that $\tr(U \cdot M_{A})\neq 0$.%
    We call the triple $(A, B, W)$ a \emph{witness} of size $n$ and we define its value to be:
    \[
        \wvalue(A, B, W):= \frac{\tr(V \cdot M_{A,B})}{\tr(U \cdot M_{A})}.
    \]
\end{definition}

We now show that the existence of a witness implies a lower bound to the LSH distortion for Ulam similarity on $\mathbb{S}_n$.

\begin{lemma}\label{lem:witness-implies-ulam-lb}
    Let $(A, B, W)$ be a witness of size $n$ (\Cref{def:ulam-witness}). Then, the LSH distortion for Ulam similarity space $(\mathbb{S}_n, \ulam)$ is at least the value of the witness, $\wvalue(A , B, W)$, even when restricted only to permutations in $A\cup B \subseteq \mathbb{S}_n$.
\end{lemma}

\begin{proof}
Consider any distribution $\cH$ over functions on $\mathbb{S}_n$ which $\Delta$-approximates $(A \cup B,\ulam)$ for $\Delta\geq 1$. For every $\pi,\sigma\in A \cup B$ we have:
\begin{equation}\label{eq:ulam-assume-lsh-prop}
    \frac{M_{A\cup B}(\pi,\sigma)}{\Delta}=\frac{\ulam(\pi,\sigma)}{\Delta} \leq P_{\cH}(\pi,\sigma) \leq \ulam(\pi,\sigma) = M_{A\cup B}(\pi,\sigma).
\end{equation}
 Note that, having fixed an ordering of the elements of $A\cup B$ that puts the elements of $A$ before the elements of $B$, $P_{\cH}$ gets naturally divided into four square blocks:
\[
    P_{\cH} = \begin{pmatrix}
        P_{A} & P_{A,B}\\
        P_{B,A} & P_B
    \end{pmatrix}.
\]

By \Cref{lem:psd-Pa}, $P_{\cH} \succeq 0$ and $W \succeq 0$, we have $\tr(WP_{\cH}) \ge 0$. Expanding blockwise, and using $V^T=V$ and $P_{B,A}^T=P_{A,B}$, we have:
\[
    \tr(U P_{A}) + \tr(U P_{B}) - 2\tr(V P_{A,B}) \ge 0.
\]
Note that, by \Cref{eq:ulam-assume-lsh-prop}, 
$P_{\cH}(\pi,\sigma) \le M_{A\cup B}(\pi,\sigma)$ and $P_{\cH}(\pi,\sigma) \geq (1/\Delta) \cdot M_{A\cup B}(\pi,\sigma)$ for all $\pi,\sigma$. Therefore, we have:
\[
    \tr(U \cdot M_{A}) + \tr(U \cdot M_{B}) \ge {2\over \Delta} \cdot \tr(V \cdot M_{A,B}).
\]
Since $M_{A}=M_{B}$, we have $\tr(U\cdot M_{A})=\tr(U \cdot M_{B})$. Hence,
\[
\Delta \ge \frac{\tr(V \cdot M_{A,B})}{\tr(U\cdot M_{A})} = \wvalue(A, B, W).
\qedhere
\]
\end{proof}

\smallskip

For permutations $\pi,\sigma\in \mathbb{S}_n$ the \emph{wreath product} $\pi \wr \sigma\in \mathbb{S}_{n^2}$ is defined as follows. The elements of $[n^2]$ are partitioned into $n$ blocks where the $i$th block, $i\in[n]$, contains elements $\Sigma_i:=\{(i-1)\cdot n + 1, \dots, i\cdot n\}$. The resulting permutation $\pi\wr \sigma$ is obtained by permuting the elements in each block with the inner permutation $\sigma$ and then permuting the blocks with the outer permutation $\pi$. For two sets $A,B\subseteq \mathbb{S}_n$, we define $A\wr B:=\{\pi\wr \sigma \mid \pi\in A, \sigma\in B\}$. We use the following key property of the wreath product for amplification:

\begin{lemma}\label{lem:ulam-wreath-prod}
    For any $\pi,\sigma,\tau_1,\tau_2 \in \mathbb{S}_n$ we have:
    \[
        \lcs(\pi\wr \sigma, \tau_1\wr \tau_2) = \lcs(\pi,
        \tau_1) \cdot \lcs(\sigma, \tau_2).
    \]
\end{lemma}
\begin{proof}
In the permutation $\pi\wr \sigma$, the blocks appear in the order given by $\pi$, and within every block the elements appear in the order given by $\sigma$; similarly for $\tau_1\wr \tau_2$.

We first prove that $\lcs(\pi\wr \sigma, \tau_1\wr \tau_2) \le \lcs(\pi,\tau_1) \cdot \lcs(\sigma, \tau_2)$.  Let $Z$ be any common subsequence of $\pi\wr \sigma$ and
$\tau_1\wr \tau_2$. Since the alphabets $\Sigma_i$ are pairwise disjoint, an element from a block
$\Sigma_i$ can only be matched with the same element from the same block $\Sigma_i$. Therefore,
if we track the sequence of block indices that contribute at least one matched element to $Z$,
then this is a common subsequence of the outer permutations $\pi$ and $\tau_1$; thus, the
number of contributing blocks is at most $\lcs(\pi,\tau_1)$.

Let $\Sigma_i$ be any contributing block. The matched elements of $Z$ in this block is
a common subsequence of the ordering of $\Sigma_i$ inside $\pi\wr \sigma$ and the ordering of $\Sigma_i$ inside $\tau_1\wr \tau_2$. These two internal orderings are $\sigma$ and $\tau_2$, relabeled with the alphabet
$\Sigma_i$; therefore, this block contributes 
at most $\lcs(\sigma,\tau_2)$ matched elements. Summing over
all contributing blocks, we obtain $\lcs(\pi\wr \sigma,\tau_1\wr \tau_2)\le \lcs(\pi,\tau_1)\cdot \lcs(\sigma,\tau_2).$

We now show $\lcs(\pi\wr \sigma, \tau_1\wr \tau_2) \ge \lcs(\pi,\tau_1) \cdot \lcs(\sigma, \tau_2)$. Consider an optimal common subsequence of the outer permutations $\pi$ and $\tau_1$, consisting of $\lcs(\pi,\tau_1)$ block indices. For each such matched block $\Sigma_i$, consider an optimal common subsequence of the inner permutations $\sigma$ and $\tau_2$, relabeled using $\Sigma_i$. Since the block indices appear in the same relative order in both $\pi$ and $\tau_1$, concatenating these common subsequences over the chosen blocks yields a valid common subsequence of $\pi\wr \sigma$ and $\tau_1\wr \tau_2$ of length $\lcs(\pi,\tau_1)\cdot \lcs(\sigma,\tau_2)$.
\end{proof}

We now show how to use the wreath product to amplify the size of the witness. We will make use of the \emph{Kronecker product}. We recall that for two matrices $X \in \mathbb{R}^{n\times m}$ and $Y\in \mathbb{R}^{s\times t}$, the Kronecker product $X \otimes Y \in\mathbb{R}^{ns \times mt}$ is the following block matrix:
\[
    X\otimes Y = \begin{pmatrix}
        x_{11} \cdot Y & \dots & x_{1m} \cdot Y \\
        \vdots & \ddots & \vdots\\
        x_{n1} \cdot Y & \dots & x_{nm} \cdot Y
    \end{pmatrix}.
\]

\begin{lemma}\label{lem:ulam-wreath-matrix}
    Consider any disjoint sets $A,B\subseteq\mathbb{S}_n$ with $|A|=|B|$. Let $A'=A\wr A$ and $B'=B\wr B$. Then, we have:
    \[
        M_{(A'\cup B')} = 
        \begin{pmatrix}
        M_{A} \otimes M_{A} & M_{A,B} \otimes M_{A,B} \\
        M_{B,A} \otimes M_{B,A} & M_{B} \otimes M_{B}
    \end{pmatrix}.
    \] 
\end{lemma}
\begin{proof}
    Let $\pi, \sigma\in A' \cup B'$ and suppose that $\pi=\pi_1 \wr \pi_2$, $\sigma=\sigma_1 \wr \sigma_2$. By \Cref{lem:ulam-wreath-prod},
    \[
        M_{(A'\cup B')}(\pi,\sigma)=\frac{\lcs(\pi,\sigma)}{n^2}=\frac{\lcs(\pi_1, \sigma_1)}{n} \cdot \frac{\lcs(\pi_2, \sigma_2)}{n} = M_{(A\cup B)}(\pi_1, \sigma_1) \cdot M_{(A\cup B)}(\pi_2, \sigma_2).
    \]
    Then, by considering the ordering over $A'$ obtained by the ordering of $A$ when we consider first the outer permutation and then the inner permutation (and similarly for $B'$), we have that each entry of the matrix is equivalent to the Kronecker product. 
\end{proof}

\begin{lemma}\label{lem:ulam-boosting}
    Let $(A, B, W)$ be a witness of size $n$. Let $A'=A\wr A$, $B'=B\wr B$, and:
    \[
        W'= 
        \begin{pmatrix}
            U\otimes U & -(V \otimes V)\\
            -(V^T \otimes V^T) & U\otimes U
        \end{pmatrix}.
    \] 
    Then, $(A', B', W')$ is a witness of size $n^2$ and $\wvalue(A', B', W') \geq \wvalue(A, B, W)^2$.
\end{lemma}
\begin{proof}
    Note that $U\otimes U$ and $V\otimes V$ are non-negative and symmetric. Moreover, since $W\succeq 0$, we have that $W\otimes W \succeq 0$. Define:
    \[
        W''= \begin{pmatrix}
            U\otimes U & V \otimes V\\
            V^T \otimes V^T & U\otimes U
        \end{pmatrix}.
    \]
    Since $W''$ is a principal sub-matrix of $W\otimes W$, we have that $W''\succeq 0$. Observe now that $W' = \text{diag}(I, -I) \cdot W'' \cdot \text{diag}(I, -I)$, and therefore $W'\succeq 0$ since for any matrix $X\succeq 0$ and any matrix $Y$, $YXY^{T} \succeq 0$. Moreover, by \Cref{lem:ulam-wreath-matrix} and since $M_{A}=M_{B}$, we also have $M_{A'}=M_{B'}$. Then, $(A', B', W')$ is a valid witness of size $n^2$. We now compute its value. 
    
    For matrices $A,B,C,D$ it holds that $\tr((A \otimes B)(C \otimes D)) = \tr(AC) \cdot \tr(BD)$. By using this property, and by \Cref{lem:ulam-wreath-matrix}, we have:
    \begin{align*}
        \tr((V\otimes V) \cdot M_{A',B'}) &= \tr((V\otimes V) \cdot (M_{A,B} \otimes M_{A,B})) = \tr(V \cdot M_{A,B})^2, \text{ and}\\
        \tr((U\otimes U) \cdot M_{A'}) &= \tr((U\otimes U) \cdot (M_{A} \otimes M_{A})) = \tr(U \cdot M_{A})^2.
    \end{align*}
    Note that since $\tr(U \cdot M_{A})\neq 0$, $\tr((U\otimes U) M_{A'})\neq 0$. Then:
    \[
        \wvalue(A', B', W') = \frac{\tr((V\otimes V) \cdot M_{A',B'})}{\tr((U\otimes U) \cdot M_{A'})} = \left(\frac{\tr(V \cdot M_{A,B})}{\tr(U \cdot M_{A})}\right)^2 = \wvalue(A, B, W)^2. \qedhere
    \]
\end{proof}

We now provide our starting instance.

\begin{proposition}\label{prop:ulam-small-instance}
    There exists a witness $(A_0, B_0, W_0)$ of size $8$ and value $\geq 9/7$.
\end{proposition}
\begin{proof}

Let $A_0 = \{A_1, \ldots, A_4\}$ and $B_0 = \{ B_1, \ldots, B_4 \}$,
where $\tau=(1\,2\,7\,8)(3\,6\,5\,4) =(2,7,6,3,4,5,8,1)$
and $B_i=\tau A_i$:
\[
\begin{aligned}
A_1 = (1,2,3,4,5,6,7,8), \qquad\qquad B_1=(2,7,6,3,4,5,8,1),\\
A_2 = (4,3,2,1,8,7,6,5), \qquad\qquad B_2=(3,6,7,2,1,8,5,4),\\
A_3 = (6,5,8,7,2,1,4,3), \qquad\qquad B_3=(5,4,1,8,7,2,3,6),\\
A_4 = (7,8,5,6,3,4,1,2), \qquad\qquad B_4=(8,1,4,5,6,3,2,7).
\end{aligned}
\]
Since $B_0$ is obtained from $A_0$ by relabeling, we have $M_{A_0}=M_{B_0}$.
A direct computation gives
\[
M_{A_0}=\frac18
\begin{pmatrix}
8&2&2&2\\
2&8&2&2\\
2&2&8&2\\
2&2&2&8
\end{pmatrix},
\qquad
M_{A_0,B_0}=\frac18
\begin{pmatrix}
5&4&4&5\\
4&5&5&4\\
4&5&5&4\\
5&4&4&5
\end{pmatrix}.
\]
The matrix $W_0$ for this instance is given by $U = V = J_4$ (where $J_4$ is the all-ones $4 \times 4$ matrix). First, 
\[
W 
= \begin{pmatrix} U & -V \\ -V^T & U \end{pmatrix}
= \left( 
\begin{array}{cc}
J_4 & -J_4 \\
-J_4 & J_4
\end{array}
\right) 
= \left( 
\begin{array}{cc}
1 & -1 \\
-1 & 1
\end{array}
\right) \otimes J_4 \succeq 0.
\]

Next, it is easy to calculate
\[
\operatorname{tr}(V \cdot M_{A_0,B_0})=\frac{1}{8}(18 \cdot 4)=9,
\qquad
\operatorname{tr}(U \cdot M_{A_0})=\frac{1}{8}(14\cdot 4)=7.
\]
Hence, $\wvalue(A_0, B_0, W_0) = 9/7$.
\end{proof}

\begin{lemma}\label{lem:ulam-lb-non-zero}
    There exists a constant $\delta>0$ such that, for any $n_0>0$, there exists $n\geq n_0$ such that the LSH distortion of the Ulam similarity $\ulam$ on $\mathbb{S}_n$ is at least $\Omega(n^\delta)$.
\end{lemma}
\begin{proof}
    Let $\Delta_0=9/7$ and $m=8$. Let $k$ be the smallest integer such that $m^{2^k} \geq n_0$, we set $n=m^{2^k}$. Note that $k=\log_2 \log_m (n)$.

    Starting from the witness of \Cref{prop:ulam-small-instance} and applying \Cref{lem:ulam-boosting} $k$ times we obtain a witness of size $m^{2^k}=n$ and value $\geq \Delta_0^{2^k}$. By \Cref{lem:witness-implies-ulam-lb}, this implies that the LSH distortion on $\mathbb{S}_n$ is at least: 
    \[
        \Delta_0^{2^k} = \Delta_0^{\log_m(n)}=n^{\log_m \Delta_0}\geq n^{0.120856}=\Omega(n^\delta),
    \]
    for $\delta = 0.12$.
\end{proof}

Note that \Cref{lem:ulam-lb-non-zero} directly implies \Cref{thm:ulam-lower-bound}. The proof is elementary.

\section{Upper Bound for Cayley Similarity}\label{sec:cayley-upper-bound}

In this section, we show that the LSH distortion for Cayley similarity is no more than $n$. 

\begin{theorem}\label{thm:cayley-upper-bound}
    The LSH distortion of $(\mathbb{S}_n, \cayley)$ is at most $n$.
\end{theorem}
\begin{proof}
    Let $\cH$ be the uniform at random distribution over the functions $h:\mathbb{S}_n \rightarrow [n]$. Clearly, for each $\pi,\sigma\in \mathbb{S}_n$, $\pi\neq \sigma$:
    \[
        \Pr_{h\sim \cH}[h(\pi)=h(\sigma)] = \frac{1}{n}.
    \]
    Since $\frac{1}{n}\leq \cayley(\pi,\sigma)\leq 1$ for all $\pi,\sigma\in \mathbb{S}_n$, we have:
    \[
        \frac{1}{n} \cdot \cayley(\pi,\sigma) \leq \Pr_{h\sim \cH}[h(\pi)=h(\sigma)] \leq \cayley(\pi,\sigma).\qedhere
    \]
\end{proof}

\section{Lower Bound for Cayley Similarity}\label{sec:cayley-lower-bound}

In this section, we show the following result:

\begin{restatable}{theorem}{CayleyLowerBound}\label{thm:cayley-lower-bound}
  The LSH distortion of $(\mathbb{S}_n,\cayley)$ is $\Omega(n)$.
\end{restatable}

The proof will rely on some standard results in the representation theory of $\mathbb{S}_n$. For convenience, we review these standard results in \Cref{app:background-representation-theory}. %

We begin with the following definition:
\begin{definition}[Bi-invariance]
    We say that a function $f:\mathbb{S}_n\times \mathbb{S}_n \to \mathbb{R} $ is \emph{bi-invariant} if:
    \[
        f(\pi, \sigma) = f(\tau_1\pi\tau_2, \tau_1\sigma\tau_2),
    \]
    for any $\pi, \sigma,\tau_1,\tau_2 \in \mathbb{S}_n$. We say a bi-invariant function is a \emph{bi-invariant PSD kernel} if, in addition, it satisfies:
    \[
        \sum_{\pi,\sigma \in \mathbb{S}_n} a_\sigma a_\pi f(\pi, \sigma) > 0,
    \]
    for every $a \in \mathbb{R}^{\mathbb{S}_n}\setminus\{0\}$. We also require that $f(\id,\id)=1$ for $f$ to be a PSD kernel.
\end{definition}

Recall that a \emph{class function} on $\mathbb{S}_n$ is a function that is constant on each conjugacy class of $\mathbb{S}_n$. We recall the following property of bi-invariant functions. (Note that equivalent results are known in the literature, and hold for more general groups--see, e.g.\ \cite[Chapter 4.5]{k08}--but we include a proof anyway for completeness).
\begin{lemma}\label{lem:existence-of-class-function}
    A function $f:\mathbb{S}_n \times \mathbb{S}_n \to \mathbb{R}$ is bi-invariant if and only if there exists a class function $g:\mathbb{S}_n \to \mathbb{R}$ such that:
    \[
        f(\sigma, \pi) = g(\sigma\pi^{-1}),
    \]
    for every $\sigma, \pi \in \mathbb{S}_n$.
\end{lemma}
\begin{proof}
    Let $f:\mathbb{S}_n \times \mathbb{S}_n \to \mathbb{R}$ be a bi-invariant function. Define $g:\mathbb{S}_n \to \mathbb{R}$ as:
    \[
        g(\pi) := f(\pi,\id).
    \]
    Then we have:
    \[
        f(\pi, \sigma) =f(\pi\sigma^{-1}, \sigma \sigma^{-1}) =f(\pi\sigma^{-1}, \id) = g(\pi\sigma^{-1})
    \]
    for any $\pi,\sigma \in \mathbb{S}_n$. Moreover, $g$ is a class function, since, for any $\pi,\sigma \in \mathbb{S}_n$:
    \[
        g(\sigma \pi \sigma^{-1}) = f(\sigma \pi \sigma^{-1},\id) = f(\pi, \sigma^{-1}\sigma) = f(\pi,\id) = g(\pi),
    \]
    where the second equality follows from the bi-invariance of $f$. 

    On the other hand, suppose that $f:\mathbb{S}_n \times \mathbb{S}_n \to \mathbb{R}$ is a function such that $f(\pi,\sigma) = g(\pi \sigma^{-1})$ for some class function $g$, then, for any $\pi, \sigma, \tau \in \mathbb{S}_n$:
    \[
        f(\pi\tau,\sigma\tau) = g(\pi \tau \tau^{-1} \sigma^{-1}) = g(\pi \sigma^{-1}) = f(\pi,\sigma),
    \]
    and:
    \[
        f(\tau\pi,\tau\sigma) = g( \tau \pi\sigma^{-1}\tau^{-1} ) = g(\pi \sigma^{-1}) = f(\pi,\sigma),
    \]
    where the second equality uses the fact that $g$ is a class function. Hence $f$ is bi-invariant.
\end{proof}

A crucial property of the Cayley similarity is that it is bi-invariant, as we now show.
\begin{lemma}\label{lem:cayley-is-bi-invariant}
    $\cayley(\cdot,\cdot)$ is bi-invariant.
\end{lemma}
\begin{proof}
    We have that, for any choice of $\pi$ and $\sigma$, $\cayley(\pi,\sigma)={c(\pi\sigma^{-1})\over n}$. Since any two conjugate permutations have the same cycle type, the function $\pi \mapsto {c(\pi)\over n}$ is a class function, and the result follows from \Cref{lem:existence-of-class-function}.
\end{proof}

\begin{definition}[Bi-invariant Distribution over Functions]
    We will say that a distribution $\cH$ over functions from $\mathbb{S}_n$ into some set $\mathcal{Y}$ is \emph{bi-invariant} if the function $P_\mathcal{H}(\cdot,\cdot)$ is bi-invariant.
\end{definition}

We now show that when studying the minimum achievable distortion for $(\mathbb{S}_n, \cayley)$, any LSH family can be assumed to be bi-invariant.
\begin{lemma}\label{lem:biinv}
Suppose there exists a distribution $\cH$ over functions from $\mathbb{S}_n$ into a finite set $\mathcal{Y}$ that $\Delta$-approximates $(\mathbb{S}_n, \cayley)$, then there exists a bi-invariant distribution $\cH'$ that $\Delta$-approximates $(\mathbb{S}_n, \cayley)$. 
\end{lemma}
\begin{proof}
We construct $\mathcal{H}'$ by pre-composing the elements of $\cH$ with random left and right permutations: define a function $h' \sim \mathcal{H}'$ as $h'(\pi) = h(\alpha \pi \beta)$, where $h$ is drawn from $\mathcal{H}$ and $\alpha, \beta$ are drawn uniformly and independently from $\mathbb{S}_n$. We first show that $\cH'$ $\Delta$-approximates $(\mathbb{S}_n, \cayley)$. For every choice of $\pi,\sigma\in\mathbb{S}_n$, we have:
\begin{align*}
\Pr_{h'\sim\mathcal{H}'}\left[h'(\pi)= h'(\sigma)\right] &= \E{\alpha, \beta \sim \mathbb{S}_n}{\Pr_{h \sim\mathcal{H}}[h(\alpha \pi \beta) = h(\alpha \sigma \beta)]} \\
&\le \E{\alpha, \beta \sim \mathbb{S}_n}{\cayley(\alpha \pi \beta, \alpha \sigma \beta)} \\
&\overset{\Cref{lem:cayley-is-bi-invariant}}{=} \E{\alpha, \beta \sim \mathbb{S}_n}{\cayley(\pi, \sigma)} = \cayley(\pi,\sigma),
\end{align*}
where the first inequality uses that $\cH$ $\Delta$-approximates $(\mathbb{S}_n,\cayley)$. Similarly:
\begin{align*}
    \Pr_{h'\sim\mathcal{H}'}\left[h'(\pi)= h'(\sigma)\right] &= \E{\alpha, \beta \sim \mathbb{S}_n}{\Pr_{h \sim\mathcal{H}}[h(\alpha \pi \beta) = h(\alpha \sigma \beta)]} \\
    &\ge \E{\alpha, \beta \sim \mathbb{S}_n}{{1\over \Delta}\cdot \cayley(\alpha \pi \beta, \alpha \sigma \beta)} \\&\overset{\Cref{lem:cayley-is-bi-invariant}}{=} \E{\alpha, \beta \sim \mathbb{S}_n}{{1\over \Delta} \cdot \cayley(\pi, \sigma)} = {1\over \Delta}\cayley(\pi,\sigma),
\end{align*}
Showing that indeed $\cH'$ $\Delta$-approximates $(\mathbb{S}_n,\cayley)$. Finally, we verify that $P_{\mathcal{H}'}$ is bi-invariant. For any $\pi,\sigma\in \mathbb{S}_n$ and any left multiplier $\tau_1$ and right multiplier $\tau_2$:
\[
    P_{\mathcal{H}'}(\tau_1 \pi \tau_2, \tau_1 \sigma \tau_2) = \mathbb{E}_{\alpha, \beta \sim \mathbb{S}_n}[P_{\mathcal{H}}(\alpha \tau_1 \pi \tau_2 \beta, \alpha \tau_1 \sigma \tau_2 \beta)].
\]
Because $\alpha$ and $\beta$ are drawn uniformly from the group, the products $\alpha' = \alpha \tau_1$ and $\beta' = \tau_2 \beta$ are also uniformly distributed. Thus, the expectation is identical to $\mathbb{E}_{\alpha', \beta'}[P_{\mathcal{H}}(\alpha' \pi \beta', \alpha' \sigma \beta')] = P_{\mathcal{H}'}(\pi, \sigma)$, completing the proof. 
\end{proof}

Note also that, by \Cref{lem:psd-Pa}, $P_{\cH'}$ is PSD and it is therefore a bi-invariant PSD kernel. 

Recall that each partition $\lambda\vdash n$ is associated with a distinct irreducible representation $S^{\lambda}$ of $\mathbb{S}_n$, we denote by $\chi_\lambda: \mathbb{S}_n \to \mathbb{C}$ its character, and by $d_\lambda\in\mathbb{N}$ its dimension. 

\begin{lemma}\label{lem:char}
Any bi-invariant PSD kernel $P$ on $\mathbb{S}_n$ can be written as a convex combination of normalized characters of irreducible representations of $\mathbb{S}_n$ as follows:
\[
   \forall \sigma,\pi\in \mathbb{S}_n: \quad  P(\sigma, \pi) = \sum_{\lambda :\lambda \vdash n} w_\lambda \frac{\chi_\lambda(\sigma\pi^{-1})}{d_\lambda}, 
\]
where $w_\lambda \ge 0$ and $\sum_\lambda  w_\lambda = 1$. 
\end{lemma}
\begin{proof}
    Let $\phi:\mathbb{S}_n \to \mathbb{R}$ be the class function satisfying: 
    \[
        \phi(\sigma\pi^{-1}) = P(\sigma, \pi),
    \]
    the existence of which is guaranteed by Lemma \ref{lem:existence-of-class-function}. By a result of \citet[Corollary 4.5.5\footnote{We remark that Kondor's result is stated in terms of what he refers to as ``Positive Definite'' kernels, which are, in reality, positive semidefinite kernels, as he himself explains earlier in the thesis.}]{k08}, since $\phi$ is PSD, we have:
    \begin{equation}\label{eq:phi-in-span-of-chi}
        \phi = \sum_{\lambda \vdash n} a_\lambda \cdot \chi_{\lambda},
    \end{equation}
    for some choice of $\{a_\lambda\}_{\lambda \vdash n}$, where each $a_\lambda \in \mathbb{R}_{\ge 0}$. 

    In particular:
    \[
        1 = P(\id,\id) = \phi(\id) = \sum_{\lambda \vdash n} a_\lambda \cdot \chi_{\lambda}(\id) = \sum_{\lambda \vdash n} a_\lambda \cdot d_\lambda,
    \]
    where we use that $\chi_\lambda(\id)=d_\lambda$. Letting $w_\lambda : = a_\lambda \cdot d_\lambda$, together with Equation \eqref{eq:phi-in-span-of-chi} completes the proof.
\end{proof}

For our asymptotic bounds, it is necessary to isolate the trivial representation $(n)$, where $\chi_{(n)}(\cdot) = 1$ and $d_{(n)} = 1$ and the sign representation $(1^n)$, where
$\chi_{(1^n)}(\sigma) = \sgn(\sigma)$ and $d_{(1^n)} = 1$.  Thus, it is helpful to
write 
\begin{equation}
\label{eq:char}
P(\id,\sigma) = w_0 + w_{\sgn} \cdot \sgn(\sigma) + \sum_{\lambda\ne(n), (1^n)}w_{\lambda}\frac{\chi_{\lambda}(\sigma)}{d_{\lambda}}.
\end{equation}

For a permutation $\sigma \in \bS_n$, the \emph{fixed-point} count is defined as $\fix(\sigma) = |\{i \in [n] \mid \sigma(i) = i\}|,$ i.e., the number of length-1 cycles.
A \emph{derangement} is a permutation $\sigma$ that moves all elements, i.e., $\fix(\sigma) = 0$. For any $\sigma \in \mathbb{S}_n$, we also define the value $\supp(\sigma)$ as $n- \fix(\sigma)$. We recall the following result.
\begin{theorem}[\cite{r96}]\label{thm:roichman}
There exist absolute constants $\varphi > 0$ and $0 < q < 1$ such that for any partition $\lambda = (\lambda_1, \dots , \lambda_k) \vdash n$, any $n \ge 4$ and any permutation $\sigma \in \mathbb{S}_n$:
\[
\left| \frac{\chi_\lambda(\sigma)}{d_\lambda} \right| \le \left( \max \left\{ \frac{\lambda_1}{n}, \frac{k}{n}, q \right\} \right)^{\varphi \cdot \supp(\sigma)}\hspace{-8mm}.
\]
\end{theorem}

\Cref{thm:roichman} allows us to prove the following.

\begin{lemma}
\label{lem:largek}
There exist absolute constants $\varphi > 0$ and $0 < q < 1$ such that for any $n\geq 4$, any partition  $\lambda = (\lambda_1, \lambda_2, \ldots ,\lambda_k) \vdash n$, and any derangement $\sigma \in \mathbb{S}_n$, 
\[
\left| \frac{\chi_\lambda(\sigma)}{d_\lambda} \right| \le e^{-\varphi\cdot \alpha},
\]
where $\alpha = \min\{n - \lambda_1, n - k, (1-q)n\}$. 
\end{lemma}
\begin{proof}
Since $\sigma$ is a derangement, $\text{fix}(\sigma) = 0$; therefore
$\supp(\sigma) = n - \fix(\sigma) = n$.  
Furthermore, from the definition of $\alpha$, we have: 
\[
\max \left\{\frac{\lambda_1}{n}, \frac{k}{n}, q \right\}
= 1 - \frac{\alpha}{n}.
\]
Substituting these into the bound
from Theorem~\ref{thm:roichman} yields:
\[
\left| \frac{\chi_\lambda(\sigma)}{d_\lambda} \right| \le \left( 1 - \frac{\alpha}{n} \right)^{\varphi n} \le e^{-\varphi\cdot \alpha}.
\qedhere
\]
\end{proof}

\begin{lemma}[Theorem 1.2 (i) from \cite{ls08}, paraphrased]\label{lem:chi-sigma-root-dimension}
    For any derangement $\sigma \in \mathbb{S}_n$ and any partition $\lambda \vdash n$:
    \[
        \left|{\chi_\lambda(\sigma) \over d_\lambda}\right| \le {1\over d_\lambda^{1/2 - o(1)}}.
    \]
\end{lemma}

\begin{lemma}\label{lem:dimension-hook-length}
    For any partition $\lambda=(\lambda_1, \dots, \lambda_k)\vdash n$ with $\lambda_1 \geq \frac{11}{12}\cdot n$, it holds that: $d_\lambda \geq \left(\frac{n}{2(n-\lambda_1)}\right)^{n - \lambda_1}$ 
\end{lemma}
\begin{proof}
    Let $\gamma=n-\lambda_1$. By the hook length formula (see \Cref{app:background-representation-theory}), we have:
    \begin{align*}
        d_\lambda = {n! \over \prod h(i,j)} &\geq \frac{n!}{(n-2\gamma)! \cdot n^\gamma \cdot \gamma^\gamma} \\
        & = \frac{n \cdots (n-2\gamma+1)}{n^\gamma \gamma^\gamma}\\
        & = \frac{n}{n} \cdots \frac{n-\gamma+1}{n} \cdot \frac{(n-\gamma) \cdots (n-2\gamma+1)}{\gamma^\gamma}\\
        &\geq \left(\frac{n-\gamma+1}{n}\right)^\gamma \cdot \left(\frac{n-2\gamma+1}{\gamma}\right)^\gamma\\
        &\geq \left(\frac{n^2-6\cdot n\cdot \gamma}{\gamma \cdot n}\right)^\gamma \\
        &=\left(\frac{n}{\gamma} - 6\right)^\gamma\\
        &=\left(\frac{n}{\gamma}\left(1-\frac{6\gamma}{n}\right)\right)^\gamma\\
        &\geq \left(\frac{n}{2\gamma}\right)^\gamma,
    \end{align*} 
    where the last inequality used the fact that $\gamma \le {n \over 12}$.
\end{proof}

\begin{lemma} \label{lem:derang}
Consider any sufficiently large $n$, any derangement $\sigma \in \mathbb{S}_n$, and let $P$ be a bi-invariant PSD kernel. Let:
\[
    P(\tau,\pi) = \sum_{\lambda\vdash n} w_\lambda \frac{\chi_\lambda(\tau\pi^{-1})}{d_\lambda},
\]
be the decomposition guaranteed by \Cref{lem:char}, then:
\[
    \left|\sum_{\substack{\lambda \vdash n\\ \lambda \neq (n), (1^n)}} w_\lambda \frac{\chi_\lambda(\sigma)}{d_\lambda}\right| \le \frac{C}{n}, 
\]
for some constant $C$. 
\end{lemma}
\begin{proof}
    We have:
    \[
        \left|\sum_{\substack{\lambda \vdash n\\ \lambda \neq (n), (1^n)}} w_\lambda \frac{\chi_\lambda(\sigma)}{d_\lambda}\right| \le  \max_{\substack{\lambda \vdash n\\\lambda \neq (n),(1^n)}}\left|\frac{\chi_\lambda(\sigma)}{d_\lambda}\right| \sum_{\substack{\lambda \vdash n\\ \lambda \neq (n), (1^n)}} w_\lambda \le \max_{\substack{\lambda \vdash n\\ \lambda \neq (n),(1^n)}}\left|\frac{\chi_\lambda(\sigma)}{d_\lambda}\right|, 
    \]
    where we first apply H\"older's inequality on $\ell_\infty$ and $\ell_1$ and then we use the fact that $\sum_{\lambda \vdash n } w_\lambda = 1$, as well as the non-negativity of the $w_\lambda$s. We now give an upper bound for the right-hand side. Let $\varphi$ and $q$ be the constants given by \Cref{lem:largek}. Let $\lambda =(\lambda_1, \ldots , \lambda_k) \vdash n$, and let $\alpha = \min \{n- \lambda_1 , n-k, (1-q)n\}$. Since $\lambda\neq(n),(1^n)$, we have $\min\{n-\lambda_1, n-k\}\geq 1$. We obtain an upper bound by considering five possible cases for the value of $\alpha$.
    \begin{description}
        \item[Case 1:] $\alpha=(1-q)n$. Using \Cref{lem:largek}, we have:
        \[
            \left|\frac{\chi_\lambda(\sigma)}{d_\lambda}\right|  \le e^{-\varphi \cdot \alpha} = e^{-\Theta(n)} = o\left({1\over n}\right).
        \]
        For the rest of the proof we will assume that $\alpha = n-\lambda_1 \geq 1$. This is without loss of generality, since, for any partition that does not satisfy this condition (i.e., for which $\alpha=n-k$), its transpose partition does satisfy it, and the values of $\left|\frac{\chi_\lambda(\sigma)}{d_\lambda}\right|$ are the same for the two partitions.%
        
        \item[Case 2:] $\alpha \ge {1\over \varphi} \ln n$. Using \Cref{lem:largek} we have:
        \[
            \left|\frac{\chi_\lambda(\sigma)}{d_\lambda}\right|  \le e^{-\varphi \cdot \alpha} \le {1\over n}.
        \]
        
        \item[Case 3:]  $3 \le \alpha < {1\over \varphi}\ln n$. Using \Cref{lem:chi-sigma-root-dimension} and \Cref{lem:dimension-hook-length} and recalling that $\alpha=n-\lambda_1$, and that therefore $\lambda_1 = n-\alpha > n - \frac{1}{\varphi}\ln n \geq \frac{11}{12}n$ for sufficiently large $n$, we obtain: 
        \[
            \left|\frac{\chi_\lambda(\sigma)}{d_\lambda}\right| \leq \left(\frac{1}{d_\lambda}\right)^{1/2-o(1)} \leq \left(\frac{2\alpha}{n}\right)^{\alpha\left(\frac{1}{2} - o(1)\right)} \leq \left(\frac{(2/\varphi )\cdot \ln(n)}{n}\right)^{3/2 -o(1)} = O\left({1\over n}\right),
        \]
        where we use that $\alpha \geq 3$ and $\alpha < O(\log n)$.
        \item[Case 4:] $\alpha=1$. In this case, the only partition is $\lambda = (n-1,1)$, corresponding to the standard representation, which has dimension $n-1$. Its character is known to satisfy $\chi_\lambda (\pi) = \fix(\pi)-1$ (See \cite[Example 2.3.8]{s01}). Since $\sigma$ is a derangement, $\fix(\sigma)=0$, and the character $\chi_\lambda(\sigma) = -1$. We then have:
        \[
             \left|\frac{\chi_\lambda(\sigma)}{d_\lambda}\right| =  \left|\frac{-1}{n-1}\right| = O\left({1\over n}\right).
        \]
        \item[Case 5:] $\alpha=2$. In this case, the only two partitions are $\lambda=(n-2, 2)$ and $\mu =(n-2, 1, 1)$.
        
        By \eqref{eq:character-decompostion-of-n-2-2} in \Cref{example:decomposition-of-reps} in \Cref{app:background-representation-theory}, we can write the character of the permutation module $M^{(n-2,2)}$ as the sum of irreducible characters as follows: %
        \[
            \chi_{M^{(n-2,2)}} = \chi_{(n-2,2)} + \chi_{(n-1,1)} + \chi_{(n)}.
        \]
        By \eqref{eq:character-of-m-n-2-2}, this equation becomes:
        \[
            \twocycle(\pi) + \binom{\fix(\pi)}{2} = \chi_{(n-2,2)}(\pi) + \fix(\pi)-1 + 1
        \]
        for all $\pi \in \mathbb{S}_n$, where $\twocycle(\pi)$ is the number of length-two cycles in the cycle decomposition of $\pi$. And hence:
        \[
            \chi_{(n-2,2)}(\pi)  =  \twocycle(\pi) + \binom{\fix(\pi)}{2} - \fix(\pi)
        \]
        for all $\pi \in \mathbb{S}_n$. 
        
        Similarly, by \eqref{eq:character-decompostion-n-2-1-1} in \Cref{example:decomposition-of-reps} in \Cref{app:background-representation-theory} we can write the character of the permutation module $M^{(n-2,1,1)}$ as:
        \[
            \chi_{M^{(n-2,1,1)}} = \chi_{(n-2,1,1)} + \chi_{(n-2,2)} + 2\chi_{(n-1,1)} + \chi_{(n)}.
        \]
        By \eqref{eq:character-of-n-2-1-1}, this equation becomes:
        \[
            \fix(\pi)(\fix(\pi) -1) = \chi_{(n-2,1,1)} + \chi_{(n-2,2)} + 2(\fix(\pi)-1) + 1
        \]
        for all $\pi\in \mathbb{S}_n$, which gives:
        \[
            \chi_{(n-2,1,1)} = 1 - \twocycle(\pi) + {\fix(\pi)(\fix(\pi)-3)\over 2}
        \]
        for all $\pi\in\mathbb{S}_n$. Since $\sigma$ is a derangement and $\fix(\sigma)=0$, we have:
        \[
            \chi_{(n-2,2)}(\sigma)  =  \twocycle(\sigma),
        \]
        and:
        \[
            \chi_{(n-2,1,1)}(\sigma) = 1 - \twocycle(\sigma).
        \]
        By \Cref{example:hook-length} in \Cref{app:background-representation-theory}: $d_\lambda = n(n-3)/2$ and $d_\mu = (n-1)(n-2)/2$. 
        Putting things together, we have that:
        \begin{align*}
            \left|\frac{\chi_\lambda(\sigma)}{d_\lambda}\right| &=\frac{2\cdot \twocycle(\sigma)}{n(n-3)} \leq \frac{n}{n(n-3)} = \frac{1}{n-3} = O\left(\frac{1}{n}\right), 
        \end{align*}
        and:
        \[
            \left|\frac{\chi_\mu(\sigma)}{d_\mu}\right| = {2\cdot (\twocycle(\sigma)-1) \over (n-1)(n-2)} \le {n \over (n-1)(n-2)} = O\left({1\over n}\right).
        \]
    \end{description}
    This completes the proof of \Cref{lem:derang}.
\end{proof}

We can now complete the proof of \Cref{thm:cayley-lower-bound}.
\begin{proof}[Proof of \Cref{thm:cayley-lower-bound}]
    Let $n\geq 6$ be a large enough integer so that \Cref{lem:derang} holds. Let $\cH$ be a distribution that $\Delta$-approximates $(\mathbb{S}_n, \cayley)$. Without loss of generality, by \Cref{lem:biinv}, we can assume that $P_\cH$ is a bi-invariant PSD kernel. 

    We let $\sigma^*$ be a derangement with $\Theta(n)$ cycles and $\pi^*$ be a derangement with $O(1)$ cycles and having the same sign as $\sigma^*$. We explicitly define these permutations expressing them in their cycle decomposition. If $n$ is even, we let:
    \[
        \sigma^*:= \prod_{i=1}^{n/2} (2i-1,2i),
    \]
    and observe that $c(\sigma^*)=n/2$. If instead $n$ is odd, we define:
    \[
        \sigma^*:= (1, 2, 3)\prod_{i=2}^{\frac{n-1}{2}} (2i,2i+1),
    \]
    and observe that $c(\sigma^*)=(n-1)/2$ and $\sigma^*$ can be decomposed in $(n+1)/2$ transpositions. Now, if $n$ is even and $\sgn(\sigma^*)=1$ or if $n$ is odd and $\sgn(\sigma^*)=-1$, we define:
    \[
        \pi^*=(1,2,\dots, n-3)(n-2, n-1, n),
    \]
    and note that $\sgn(\pi^*)=\sgn(\sigma^*)$ since it can be decomposed into $n-2$ transpositions. If instead $n$ is even and $\sgn(\sigma^*)=-1$ or if $n$ is odd and $\sgn(\sigma^*)=1$, we define:
    \[
        \pi^* = (1, 2, \dots, n),
    \]
    and note that $\sgn(\pi^*)=\sgn(\sigma^*)$ since it can be decomposed into $n-1$ transpositions. 
    
    Observe that, since $n\geq 6$:
    \[
    \cayley(\id, \sigma^*) \geq \frac{(n-1)/2}{n} \geq \frac{1}{3}\hspace{1cm} \text{and} \hspace{1cm}\cayley(\id, \pi^*) \leq {2\over n}.
    \]
    Since $\cH$ $\Delta$-approximates $(\mathbb{S}_n, \cayley)$, we have:
    \[
        P_{\cH}(\id, \pi^*) \le \cayley(\id, \pi^*) \leq {2\over n}.
    \]
    Plugging this into \eqref{eq:char}, and letting $C$ be the constant in the statement of \Cref{lem:derang}:
    \begin{align*}
        w_0 + w_{\sgn} - {C\over n} &\overset{\Cref{lem:derang}}{\le} w_0 + w_{\sgn} \cdot \sgn(\pi^*) + \sum_{\substack{\lambda \vdash n \\\lambda\ne(n), (1^n)}}w_{\lambda}\frac{\chi_{\lambda}(\pi^*)}{d_{\lambda}}
        \overset{\eqref{eq:char}}{=} P_{\cH}(\id, \pi^*)
        \le {2\over n},
    \end{align*}
    giving:
    \begin{equation}\label{eq:ub-some-of-weights}
        w_0 + w_{\sgn} \le {2+C\over n}.
    \end{equation}
    Again applying \Cref{eq:char} and \Cref{lem:derang}, we have:
    \[
        P_{\cH}(\id, \sigma^*) \overset{\eqref{eq:char}, \Cref{lem:derang}}{\le} w_0 + w_{\sgn} +{C \over n} \overset{\eqref{eq:ub-some-of-weights}}{\le} {2 + 2C \over n}.
    \]
    On the other hand: 
    \[
        P_{\cH}(\id, \sigma^*) \ge {1\over \Delta} \cayley(\id, \sigma^*) = {1\over 3\Delta}.
    \]
    Combining the last two equations yields:
    \[
        \Delta \ge {n\over 6(1+C)},
    \]
    completing the proof.
\end{proof}

\section{Conclusions and Open Questions}

Our work shows that the Ulam similarity enjoys sublinear
LSH distortion whereas the Cayley similarity does not.  An
obvious open problem is to close the glaring yet tantalizing 
gap between the
upper bound ($n/\sqrt{\log n}$) and the lower bound ($\Omega(n^{0.12})$) for the Ulam similarity; our
current belief is that the correct LSH distortion bound
is $\tilde{\Theta}(\sqrt{n})$.  Improving the base case
(\Cref{prop:ulam-small-instance}) is a natural direction,
though there are computational limitations to finding
a larger and better base instance.  While it seems possible to improve our upper bound to $\tilde{O}(n/\log n)$, going beyond this to $\tilde{O}(n^{1 - \Omega(1)})$ might  
require an entirely new approach.
Extending our results to other permutation similarities
such as reversal is another interesting direction; our
current lower bound techniques seem inadequate for the block reversal similarity, for which
we conjecture an $\Omega(n)$ lower bound.

\section*{Acknowledgments}
The authors wish to thank Alessandro Panconesi for providing useful ideas for the project.

The Ulam upper bound was obtained after weeks of
interaction with and help from an LLM; the LLM-generated
proof was hand-verified and somewhat simplified.  
The solver for the base case of the Ulam
lower bound was written by an LLM; the instance itself
was hand-verified.  The high-level approach of using representation theory, in particular,
Roichman's Theorem and derangements, to obtain the Cayley lower bound was suggested by an LLM (after the initial attempts to extend Ulam lower bound proof approach to Cayley failed). All the final proofs were written and verified by the authors.
\bibliographystyle{plainnat}
\bibliography{bib}

\newpage
\appendix
\crefalias{section}{appendix}

\section{Background: The Symmetric Group and its Representations}\label{app:background-representation-theory}

In this section, we review some basic results in the representation theory of symmetric groups. This is not intended to be a comprehensive introduction to the subject, and we refer the reader to some references (e.g., \cite{s01,k08}) for a more detailed treatment. In the discussion that follows every group is assumed to be finite.

\paragraph{Group Theory of $\mathbb{S}_n$.} Two elements $x,y$ of a group $G$ are conjugate to each other if there exists some element $z\in G$ such that $zxz^{-1}=y$. This defines an equivalence relation and the corresponding classes are called \emph{conjugacy classes} of $G$. A \emph{class function} is a function that is constant on each equivalence class. $\mathbb{S}_n$ is the group of permutations of $[n]$ under composition.
Any permutation in $\mathbb{S}_n$ can be written as the composition of disjoint cycles in a way that is unique up to changing their order. This representation is known as the \emph{cycle decomposition} of the permutation. The action of an element $\sigma$ by conjugation on an element $\pi$ simply re-labels the elements of $[n]$ according to $\sigma$. A key implication of this fact, is that the conjugacy classes of $\mathbb{S}_n$ are labeled by the distinct shapes of cycle decompositions (the \emph{cycle types}), and are in one-to-one correspondence with the partitions of $n$. It is also known that any permutation can be written as the product of transpositions: permutations that simply swap two elements $i$ and $j$. While this representation is not unique, the number of transpositions in each of these factorizations always has the same parity. This allows us to define the sign $\sgn (\pi)$ of a permutation $\pi$, as $1$ if the permutation can be written as the product of an even number of transpositions, and $-1$ otherwise. Given a permutation $\pi \in\mathbb{S}_n$, we denote by $\fix(\pi)$ the number of fixed points of $\pi$, and by $\twocycle(\pi)$ the number of cycles of length $2$ in the cycle decomposition of $\pi$.

\paragraph{Representations of Finite Groups.}A matrix representation of a finite group $G$ is a homomorphism $\rho :G \to \mathrm{GL}_d$ where $\mathrm{GL}_d$ is the group of $d\times d$ invertible matrices in $\mathbb{C}^{d\times d}$. Given a complex vector space $V$ of dimension $d$, we let $GL(V)$ be the set of invertible linear transformation on $V$. Clearly, $GL(V)$ and $GL_d$ are isomorphic, and a representation can be then thought of as a homomorphism of $G$ into $GL(V)$. Under this equivalence, a representation is equivalent to a $G$-Module: a vector space $V$ on which $G$ acts linearly, in a way that preserves the group operation.
A subrepresentation is a subspace (/submodule) $W \subseteq V$ with the property that $\rho(G)\cdot W \subseteq W$, i.e., $W$ is closed under the linear action of elements of $G$. A representation $V$ is irreducible if it contains no non-trivial subrepresentation, i.e. no subrepresentation other than $\{0\}$ and $V$ itself.

Two representations $\rho: G \to GL_d$ and $\theta: G \to GL_d$ are equivalent (or isomorphic) if there exists some invertible matrix $T$ with the property that $T\cdot \theta(g) = \rho(g)\cdot T$ for every $g \in G$.
The character $\chi$ of a representation $\rho$ is the operator $\chi:G \to \mathbb{C}$ given by $\chi(g) = \tr (\rho(g))$. A consequence of these definitions, is that for any representation $\rho:G \to \mathrm{GL}_d$, its character $\chi$ satisfies $\chi(e_G) = d$, where $e_G$ is the identity element in $G$. A simple argument shows that characters are class functions, and, in fact, it is known that the characters of the non-equivalent irreducible representations of a finite group $G$ form a basis for the space of class functions. A consequence of this fact is that the number of such characters is the same as the number of conjugacy classes of $G$.

\paragraph{Young Diagrams, Tableaux and Tabloids.} A \emph{Young diagram} is a visual, geometric way to represent an integer partition. Given a partition $\lambda =(\lambda_1, \ldots , \lambda_r)$, a Young diagram of shape $\lambda$ is a shape containing $r$ left-aligned rows of congruent squares, where the $i^{th}$ row contains $\lambda_i$ squares. The reader might find the following example, showing the Young diagram for $\lambda =(5,5,2,2,1)$, more illuminating than the definition:

\begin{figure}[H]
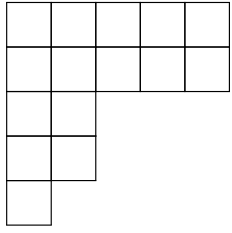

    \centering
    \ydiagram{5,5,2,2,1}
    \caption{The Young diagram for the partition $\lambda =(5,5,2,2,1)$.}
    \label{fig:young-diagram}
\end{figure}

A \emph{Young tableau} for a partition $\lambda\vdash n$, sometimes called a $\lambda$-tableau or a tableau of shape $\lambda$, is a Young diagram for $\lambda$ in which every square has been filled by a distinct number in $[n]$. 

\begin{figure}[H]
    \centering
    \begin{ytableau}
    11 & 8 & 3 & 10 & 14 \\
    13 & 7 & 2 & 9 \\
    6 & 4 \\
    12 & 1 \\
    5
    \end{ytableau}
    \caption{A Young tableau for the partition $\lambda =(5,4,2,2,1)$.}
    \label{fig:tableau}
\end{figure}

If $\lambda \vdash n$, then $\mathbb{S}_n$ has a natural action on the set of $\lambda$-tableaux: given a tableau $T$ the tableau $\pi \cdot T$ is obtained by replacing each number $i$ with $\pi(i)$.

A \emph{Young tabloid} of shape $\lambda$ is an equivalence class of $\lambda$-tableaux under the equivalence relation in which two tableaux are equivalent if you can obtain one from the other by permuting the content of each row. These are typically depicted by removing the vertical lines from one of the tableaux in the class. For example, a tabloid of shape $\lambda = (2,2,1)\vdash n$ is:

\ytableausetup{tabloids}

\[
    \begin{ytableau}
        1 & 3 \\
        5 & 4\\
        2
    \end{ytableau}
    \ytableausetup{notabloids}
     = \left\{\hspace{4mm}
     \begin{ytableau}
        1 & 3 \\
        4 & 5\\
        2
    \end{ytableau}\;,\hspace{4mm}
    \begin{ytableau}
        3 & 1 \\
        4 & 5\\
        2
    \end{ytableau}\;,\hspace{4mm}
    \begin{ytableau}
        1 & 3 \\
        5 & 4\\
        2
    \end{ytableau}\;,\hspace{4mm}
    \begin{ytableau}
        3 & 1 \\
        5 & 4\\
        2
    \end{ytableau}\hspace{4mm}
    \right\}
\]
The group $\mathbb{S}_n$ acts naturally on tabloids too. The action of a permutation $\pi$ on a tabloid is found by choosing an arbitrary tableau $t$ in the tabloid, taking the action of the permutation of that tableau to find a new tableau $t' = \pi \cdot t$, and finally returning the tabloid containing $t'$. The reader may want to verify that this operation is indeed well-defined, and does not depend on the choice of the representative $t$. Given a partition $\lambda \vdash n$, the permutation module $M^{\lambda}$ corresponding to $\lambda$ is the vector space $\mathbb{C}\{\{t_1\}, ... , \{t_r\}\}$ of formal complex linear combinations of the complete set of distinct tabloids $\{t_1\} \ldots \{t_r\}$ of shape $\lambda$. The action of $\mathbb{S}_n$ on tabloids of shape $\lambda \vdash n$ extends linearly to $M^{\lambda}$, giving a representation of $\mathbb{S}_n$.

A \emph{generalized semistandard Young tableau} of shape $\lambda = (\lambda_1, \ldots , \lambda_r )\vdash n$ and content $\mu =(\mu_1, \ldots, \mu_\ell)$ where $\mu$ is a composition of $n$, is a Young diagram of shape $\lambda$ in which each square is filled with a number, satisfying the following properties:
\begin{enumerate}
    \item[(a)] For each $i\in[\ell]$, the number $i$ appears $\mu_i$ times in the diagram,
    \item[(b)] The numbers in each row are weakly increasing,
    \item[(c)] The numbers in each columns are strictly increasing.
\end{enumerate}
The number of generalized semistandard Young tableaux of shape $\lambda$ and content $\mu$ is the so-called \emph{Kostka number} $K_{\lambda\mu}$.

\paragraph{Representations of $\mathbb{S}_n$.} Non-equivalent, irreducible representations of $\mathbb{S}_n$ are in one-to-one correspondence with partitions $\lambda \vdash n$. Each of these $\lambda$ corresponds to the representation $S^\lambda$, sometimes called the $\lambda$-Specht module, and its character is denoted by $\chi_\lambda$. (The construction of these modules is beyond the scope of this paper, and the curious reader is advised to refer to \cite{s01}.) Under this correspondence, $\lambda = (n)$ corresponds to the trivial representation, in which every group element gets mapped to the identity matrix, and $\chi_{(n)}(\pi) =1$ for every $\pi \in \mathbb{S}_n$. On the other hand, the partition $(1^n)$ corresponds to the sign representation, and satisfies $\chi_{(1^n)}(\pi) =\sgn (\pi)$ for every $\pi \in \mathbb{S}_n$. It is known that (see e.g. \cite{r23}), for any permutation $\sigma \in \mathbb{S}_n$ and any partition $\lambda$: $|\chi_\lambda(\sigma)| =|\chi_{\lambda'}(\sigma)|$ (in fact, the values of $\chi_\lambda(\sigma)$ and $\chi_{\lambda'}(\sigma)$ are related by a factor of $\sgn (\sigma)$, but we will not need this stronger fact). Moreover $d_\lambda = d_{\lambda'}$.

Given the Young diagram for a $\lambda =(\lambda_1, ..., \lambda_k)$, the \emph{arm} of a square $(i,j)$ in the diagram, is the collection of squares on the same row that are to the right of $(i,j)$, while the \emph{leg} of $(i,j)$ is the collection of squares on the same column that are below it. The hook length of $(i,j)$ is the value $h_{\lambda}(i,j)$ given by the sum of the cardinalities of the arm and the leg, plus one. Consider the following example for the partition $\lambda = (4,3,1)$:

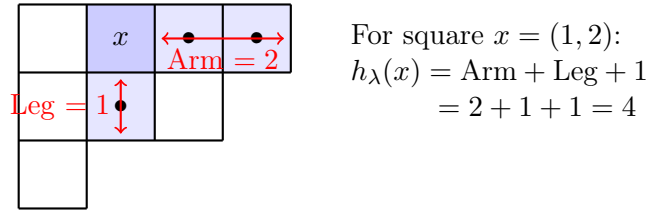
\begin{center}
\begin{tikzpicture}[scale=0.9]
    \fill[blue!20] (1,0) rectangle (2,-1); %
    \fill[blue!10] (2,0) rectangle (4,-1); %
    \fill[blue!10] (1,-1) rectangle (2,-2); %

    \draw[thick] (0,0) -- (4,0);
    \draw[thick] (0,-1) -- (4,-1);
    \draw[thick] (0,-2) -- (3,-2);
    \draw[thick] (0,-3) -- (1,-3);

    \draw[thick] (0,0) -- (0,-3);
    \draw[thick] (1,0) -- (1,-3);
    \draw[thick] (2,0) -- (2,-2);
    \draw[thick] (3,0) -- (3,-2);
    \draw[thick] (4,0) -- (4,-1);

    \node at (1.5, -0.5) {$x$};
    \node at (2.5, -0.5) {$\bullet$};
    \node at (3.5, -0.5) {$\bullet$};
    \node at (1.5, -1.5) {$\bullet$};

    \draw[<->, thick, red] (2.1, -0.5) -- (3.9, -0.5) node[midway, below] {Arm $= 2$};
    \draw[<->, thick, red] (1.5, -1.1) -- (1.5, -1.9) node[midway, left] {Leg $= 1$};

    \node[anchor=west] at (4.5, -1) {
        \begin{tabular}{l}
        \text{For square $x=(1,2)$:} \\
        $h_\lambda(x) = \text{Arm} + \text{Leg} + 1$ \\
        $\quad\quad \quad= 2 + 1 + 1 = 4$
        \end{tabular}
    };
\end{tikzpicture}
\end{center}

The dimension $d_\lambda$ of the irreducible representation $S^\lambda$ is given by the so-called \emph{hook length} formula:
\[
    d_\lambda := {n!\over \prod_{i,j} h_\lambda(i,j)}.
\]
In the following Young diagram for the partition $\lambda =(4,3,1)\vdash 8$, we have placed the value of $h_\lambda(i,j)$ in each square $(i,j)$. We can then compute the dimension $d_\lambda$ by dividing $8!$ by the product of all these numbers:

\begin{center}
\begin{tikzpicture}[scale=0.9]
    \draw[thick] (0,0) -- (4,0);
    \draw[thick] (0,-1) -- (4,-1);
    \draw[thick] (0,-2) -- (3,-2);
    \draw[thick] (0,-3) -- (1,-3);

    \draw[thick] (0,0) -- (0,-3);
    \draw[thick] (1,0) -- (1,-3);
    \draw[thick] (2,0) -- (2,-2);
    \draw[thick] (3,0) -- (3,-2);
    \draw[thick] (4,0) -- (4,-1);

    \node at (0.5, -0.5) {6};
    \node at (1.5, -0.5) {4};
    \node at (2.5, -0.5) {3};
    \node at (3.5, -0.5) {1};
    \node at (0.5, -1.5) {4};
    \node at (1.5, -1.5) {2};
    \node at (2.5, -1.5) {1};
    \node at (0.5, -2.5) {1};

    \node[anchor=west] at (4.5, -1.5) {
        $ d_\lambda = \dim(S^\lambda) = \frac{8!}{6 \cdot 4 \cdot 3 \cdot 1 \cdot 4 \cdot 2 \cdot 1 \cdot 1} = \frac{40320}{576} = 70 $
    };
\end{tikzpicture}
\end{center}

We apply the same procedure in the following examples, to calculate the dimension of some irreducible representations that are used in the body of the paper.

\begin{example}\label{example:hook-length}
    For any $n \ge 4$ we can compute the dimension of various irreducible representation.

    For $\lambda = (n)$
\begin{center}
\begin{tikzpicture}[scale=0.9]
    \draw[thick] (0,0) -- (3.5,0);
    \draw[thick] (0,-1) -- (3.5,-1);
    \draw[thick] (0,0) -- (0,-1);
    \draw[thick] (1,0) -- (1,-1);
    \draw[thick] (2,0) -- (2,-1);
    \draw[thick] (3,0) -- (3,-1);

    \node at (3.8, -0.5) {$\dots$};

    \draw[thick] (4,0) -- (5.5,0);
    \draw[thick] (4,-1) -- (5.5,-1);
    \draw[thick] (4.5,0) -- (4.5,-1);
    \draw[thick] (5.5,0) -- (5.5,-1);

    \node at (0.5, -0.5) {$n$};
    \node at (1.5, -0.5) {\small $n\!-\!1$}; %
    \node at (2.5, -0.5) {\small $n\!-\!2$};
    \node at (5, -0.5) {$1$};

\end{tikzpicture}
\end{center}

    \begin{equation}
        d_{(n)} = {n! \over n!} = {1}
    \end{equation}

    For $\lambda = (n-1,1)$:
    \begin{center}
    \begin{tikzpicture}[scale=0.9]
    \draw[thick] (0,0) -- (3.5,0);
    \draw[thick] (0,-1) -- (3.5,-1);

    \draw[thick] (0,0) -- (0,-1);
    \draw[thick] (1,0) -- (1,-1);
    \draw[thick] (2,0) -- (2,-1);
    \draw[thick] (3,0) -- (3,-1);

    \draw[thick] (0,-1) -- (0,-2);
    \draw[thick] (1,-1) -- (1,-2);
    \draw[thick] (0,-2) -- (1,-2);

    \node at (3.9, -0.5) {$\dots$};

    \draw[thick] (4.3,0) -- (5.5,0);
    \draw[thick] (4.3,-1) -- (5.5,-1);
    \draw[thick] (4.5,0) -- (4.5,-1);
    \draw[thick] (5.5,0) -- (5.5,-1);

    \node at (0.5, -0.5) {$n$};

    \node at (1.5, -0.5) {\small $n\!-\!2$}; 
    \node at (2.5, -0.5) {\small $n\!-\!3$};
    \node at (5, -0.5) {$1$};

    \node at (0.5, -1.5) {$1$};

\end{tikzpicture}
\end{center}
    
\begin{equation}
    d_{(n-1,1)} = {n! \over  (n-2)! \cdot n} = {n-1}
\end{equation}
For $\lambda =(n-2,1,1)$:

\begin{center}
\begin{tikzpicture}[scale=0.9]
    \draw[thick] (0,0) -- (3.5,0);
    \draw[thick] (0,-1) -- (3.5,-1);

    \draw[thick] (0,0) -- (0,-1);
    \draw[thick] (1,0) -- (1,-1);
    \draw[thick] (2,0) -- (2,-1);
    \draw[thick] (3,0) -- (3,-1);

    \draw[thick] (0,-1) -- (0,-3);
    \draw[thick] (1,-1) -- (1,-3);
    \draw[thick] (0,-2) -- (1,-2);
    \draw[thick] (0,-3) -- (1,-3);

    \node at (3.9, -0.5) {$\dots$};

    \draw[thick] (4.3,0) -- (5.5,0);
    \draw[thick] (4.3,-1) -- (5.5,-1);
    \draw[thick] (4.5,0) -- (4.5,-1);
    \draw[thick] (5.5,0) -- (5.5,-1);

    \node at (0.5, -0.5) {$n$};

    \node at (1.5, -0.5) {\small $n\!-\!3$}; 
    \node at (2.5, -0.5) {\small $n\!-\!4$};
    \node at (5, -0.5) {$1$};

    \node at (0.5, -1.5) {$2$};

    \node at (0.5, -2.5) {$1$};

\end{tikzpicture}
\end{center}

\begin{equation}
    d_{(n-2,1,1)} = {n! \over 2 \cdot (n-3)! \cdot n} ={(n-1)(n-2)\over 2},
\end{equation}
    
For $\lambda =(n-2,2)$:
\begin{center}
\begin{tikzpicture}[scale=0.9]
    \draw[thick] (0,0) -- (3.5,0);
    \draw[thick] (0,-1) -- (3.5,-1);

    \draw[thick] (0,0) -- (0,-1);
    \draw[thick] (1,0) -- (1,-1);
    \draw[thick] (2,0) -- (2,-1);
    \draw[thick] (3,0) -- (3,-1);

    \draw[thick] (0,-2) -- (2,-2);
    \draw[thick] (0,-1) -- (0,-2);
    \draw[thick] (1,-1) -- (1,-2);
    \draw[thick] (2,-1) -- (2,-2);

    \node at (4, -0.5) {$\dots$};

    \draw[thick] (4.3,0) -- (5.5,0);
    \draw[thick] (4.3,-1) -- (5.5,-1);
    \draw[thick] (4.5,0) -- (4.5,-1);
    \draw[thick] (5.5,0) -- (5.5,-1);

    \node at (0.5, -0.5) {\small $n\!-\!1$};

    \node at (1.5, -0.5) {\small $n\!-\!2$}; 

    \node at (2.5, -0.5) {\small $n\!-\!4$};

    \node at (5, -0.5) {$1$};

    \node at (0.5, -1.5) {$2$};

    \node at (1.5, -1.5) {$1$};

\end{tikzpicture}
\end{center}
\begin{equation}
    d_{(n-2,2)} = {n! \over 2 \cdot (n-4)! \cdot (n-2)\cdot(n-1)} = {n(n-3)\over 2},
\end{equation}
\end{example}

For any partition $\lambda =(\lambda_1, \ldots , \lambda_r) \vdash n$, the dimension of the permutation module $M^{\lambda}$ is the number of tabloids of shape $\lambda$, i.e.:
\[
    \operatorname{dim} M^{\lambda} = {n!\over \prod_{i\in [r]} \lambda_i!}.
\]
The group $\mathbb{S}_n$ acts on $M^{\lambda}$ by permuting its basis vectors, and hence, in the standard basis, the actions of the elements of $\mathbb{S}_n$  are all represented by permutation matrices. As a result, for every $\pi \in \mathbb{S}_n$ the character $\chi_{M^{\lambda}}(\pi)$ equals the number of tabloids of shape $\lambda$ fixed by the action of $\pi$. From this, we can infer, for example, that for every $\pi \in \mathbb{S}_n$:
\begin{equation}\label{eq:character-of-m-n-2-2}
    \chi_{M^{(n-2,2)}}(\pi) = \twocycle(\pi) + \binom{ \fix(\pi)}{2},
\end{equation}
\begin{equation}\label{eq:character-of-n-2-1-1}
    \text{and }\quad \chi_{M^{(n-2,1,1)}}(\pi) = \fix(\pi)(\fix(\pi)-1).
\end{equation}

By Maschke's Theorem, each representation $M^{\mu}$, decomposes as the direct sum of irreducible representations. The number of times an irreducible representation $S^{\lambda}$ appears in this sum is given by $K_{\lambda \mu}$. We illustrate the usefulness of this fact in the following example.

\begin{example}\label{example:decomposition-of-reps}
    A simple combinatorial argument shows that:
    \[
        K_{(n)(n-2,2)} = K_{(n-1,1)(n-2,2)} =K_{(n-2,2)(n-2,2)} =1.
    \]
    As a consequence of this fact, we can see that the representations $S^{(n-2,2)}$, $S^{(n-1,1)}$ and $S^{(n)}$ each appear exactly once inside of $M^{(n-2,2)}$. Since:
\[
    \dim M^{(n-2,2)} = \binom{n}{2} = d_{(n-2,2)} + d_{(n-1,1)} + d_{(n)},
\]
we find that:
\[
    M^{(n-2,2)}\cong S^{(n-2,2)} \oplus S^{(n-1,1)} \oplus S^{(n)} 
\]
and hence:
\begin{equation}\label{eq:character-decompostion-of-n-2-2}
    \chi_{M^{(n-2,2)}} = \chi_{(n-2,2)} + \chi_{(n-1,1)} + \chi_{(n)}.
\end{equation}
Similarly, we have:
\[
    K_{(n)(n-2,1,1)} = K_{(n-2,1,1)(n-2,1,1)} =K_{(n-2,2)(n-2,1,1)} =1,
\]
and:
\[
    K_{(n-1,1)(n-2,1,1)} = 2.
\]
This shows that each of the representations $S^{(n)}$, $S^{(n-2,1,1)}$ and $S^{(n-2,2)}$ appears exactly once within $M^{(n-2,1,1)}$ while $S^{(n-1,1)}$ appears twice. Since:
\[
    \operatorname{dim} M^{(n-2,1,1)} = n(n-1) =d_{(n)} + d_{(n-2,1,1)} + d_{(n-2,2)} + 2d_{(n-1,1)}
\]
we get that:
\[
    M^{(n-2,1,1)} \cong S^{(n)} \oplus S^{(n-2,1,1)} \oplus S^{(n-2,2)} \oplus S^{(n-1,1)} \oplus S^{(n-1,1)},
\]
giving:
\begin{equation}\label{eq:character-decompostion-n-2-1-1}
    \chi_{M^{(n-2,1,1)}} = \chi_{(n)} + \chi_{(n-2,1,1)} + \chi_{(n-2,2)} + 2\chi_{(n-1,1)}.
\end{equation}

\end{example}

\end{document}